%% file: main_arxiv.tex
\documentclass[sigconf, nonacm, pdfa]{acmart}
\usepackage{placeins}

\usepackage[linesnumbered,ruled,vlined]{algorithm2e}
\usepackage{xcolor}
\usepackage[normalem]{ulem}
\usepackage{todonotes}
\usepackage{etoolbox} 
\usepackage{tabularx} 
\usepackage{enumitem}
\usepackage{setspace}
\usepackage{subfig}
\usepackage{makecell}
\usepackage{caption}
\usepackage{verbatimbox}
\usepackage{float}
\usepackage{multirow}
\usepackage{xcolor-material}
\usepackage[most]{tcolorbox}
\usepackage{nameref}
\usepackage{hyperref}

\usepackage{tikz}
\usetikzlibrary{positioning,shapes,arrows,calc,fit}
\usepackage{tikz-uml}

\colorlet{amaranth}{MaterialRed600}

\SetKwComment{Comment}{\color{MaterialGreen700}// }{}
\newcommand{\inlinetcp}[1]{\textcolor{MaterialGreen700}{\texttt{// #1}}}

\SetKwProg{Fn}{Procedure}{}{}
\SetKw{Continue}{continue}

\DontPrintSemicolon
\newcommand{\boxit}[3][1]{  
    \tikz[remember picture,overlay] \node (A) {};\ignorespaces
    \tikz[remember picture,overlay]{
        \node[
            yshift=3pt,
            fill=#2,
            rounded corners=2pt,
            line width=0.5pt,
            fit={($(A)+(0,0.15\baselineskip)$)
                 ($(A)+(.9\linewidth/#1,-#3\baselineskip - \baselineskip)$)}
        ] {};
    }\ignorespaces
}

\newcommand\vldbdoi{XX.XX/XXX.XX}

\newcommand\vldbauthors{\authors}
\newcommand\vldbtitle{\shorttitle} 
\newcommand\vldbavailabilityurl{https://doi.org/10.5281/zenodo.15593109}

\newcommand{\algoname}{Cuckoo Heavy Keeper}

\newcommand{\paralgonamei}{mCHK-I}
\newcommand{\paralgonameq}{mCHK-Q}  
\newcommand{\algoemph}[1]{\emph{#1}}
\newcommand{\Thp}{L}

\newif\ifdebug 
\debugtrue 

\newcommand{\vinhcom}[1]{\todo[color=purple!40]{Vinh: #1}}
\newcommand{\vinhfixed}[2]{%
  \ifx\relax#1\relax%
    \textcolor{blue}{#2}%
  \else%
    \vinhcom{\sout{#1}}\textcolor{blue}{#2}%
  \fi%
}
\newcommand{\vinhreadmpcom}[1]{\todo[color=orange!40]{(read) MP: #1}}
\newcommand{\vinhreadmpcomfoot}[1]{\vinhreadmpcom{#1}}
\newcommand{\vinhfixedmpcom}[2]{%
  \todo[color=purple!40]{Marina: \sout{#1}}%
  \IfValueTF{#2}{\textcolor{green}{#2}}{}%
}

\newcommand{\vinhrevise}[2]{%
  \ifx\relax#1\relax%
    \textcolor{blue}{#2}%
  \else%
    \vinhcom{\sout{#1}}\textcolor{blue}{#2}%
  \fi%
}
\newcommand{\vinhemphrevise}[2]{%
  \ifx\relax#1\relax%
    \textcolor{red}{#2}%
  \else%
    \vinhcom{\sout{#1}}\textcolor{red}{#2}%
  \fi%
}
\newcommand{\vinhsectionrevise}[1]{{\color{blue}#1}}

\newcommand{\mpcom}[1]{\todo[color=orange!40]{MP: #1}}
\newcommand{\mpcomfoot}[1]{\todo[color=orange!40]{MP: #1}}

\newcommand{\mprevise}[2]{%
  \ifx\relax#1\relax%
    \textcolor{red}{#2}%
  \else%
    \mpcom{\sout{#1}}\textcolor{red}{#2}%
  \fi%
}

\newcommand{\setreleasemargins}{%
}
 
\newcommand{\setreleasecommands}{%
  \renewcommand{\vinhcom}[1]{}
  \renewcommand{\vinhfixed}[2]{##2}
  \renewcommand{\mpcom}[1]{}
  \renewcommand{\mpcomfoot}[1]{}
  \renewcommand{\vinhreadmpcom}[1]{}
  \renewcommand{\vinhreadmpcomfoot}[1]{}
  \renewcommand{\vinhfixedmpcom}[2]{##2}
  \renewcommand{\vinhrevise}[2]{##2}
  \renewcommand{\vinhemphrevise}[2]{\textcolor{red}{##2}}
  \renewcommand{\vinhsectionrevise}[1]{{\color{blue}##1}}
}

\newcommand{\releasemode}{%
  \debugfalse
  \setreleasemargins
  \setreleasecommands
}

\releasemode
 
\begin{document}
\title{\algoname{} and the balancing act of maintaining heavy hitters in stream processing}
\author{Vinh Quang Ngo}
\orcid{0009-0005-6638-9922}
\affiliation{
  \institution{Chalmers Un. of Technology and Gothenburg Un.}
  \city{Gothenburg}
  \country{Sweden}
}
\email{vinhq@chalmers.se}

\author{Marina Papatriantafilou}
\orcid{0000-0001-9094-8871}
\affiliation{
  \institution{Chalmers Un. of Technology and Gothenburg Un.}
  \city{Gothenburg}
  \country{Sweden}
}
\email{ptrianta@chalmers.se}

\begin{abstract}
Finding heavy hitters in databases and data streams is a fundamental problem with applications ranging from network monitoring to database query optimization, machine learning, and more. 
Approximation algorithms offer practical solutions, but they present trade-offs involving throughput, memory usage, and accuracy. Moreover, modern applications further complicate these trade-offs by demanding capabilities beyond sequential processing that require both parallel scaling and support for concurrent queries and updates.

Analysis of these trade-offs led us to the key idea behind our proposed streaming algorithm, \emph{\algoname{}} (CHK). The approach introduces an inverted process for distinguishing frequent from infrequent items,
which unlocks new algorithmic synergies that were previously inaccessible with conventional approaches.
By further analyzing the competing metrics with a focus on parallelism, we propose an algorithmic framework that balances scalability aspects and provides options to optimize query and insertion efficiency based on their relative frequencies. The framework is capable of parallelizing any heavy-hitter detection algorithm. 

Besides the algorithms' analysis, we present an extensive evaluation on both real-world and synthetic data across diverse distributions and query selectivity, representing the broad spectrum of application needs. 
Compared to state-of-the-art methods, CHK improves throughput by 1.7-5.7$\times$ and accuracy by up to four orders of magnitude even under low-skew data and tight memory constraints. 
These properties allow its parallel instances to achieve near-linear scale-up and low latency for heavy-hitter queries, even under a high query rate.
We expect the versatility of CHK and its parallel instances to impact a broad spectrum of tools and applications in large-scale data analytics and stream processing systems.

\end{abstract}

\maketitle

\pagestyle{plain}
\begingroup\small\noindent\raggedright\textbf{Reference Format:}\\
\vldbauthors. \vldbtitle. 
Available at: \href{https://arxiv.org/abs/XXXX.XXXXX}{arXiv:XXXX.XXXXX}.
\endgroup

\begingroup
\renewcommand\thefootnote{}\footnote{\noindent
This work is made available under the Creative Commons BY-NC-ND 4.0 International License. Visit \url{https://creativecommons.org/licenses/by-nc-nd/4.0/} for details. For any use beyond those covered by this license, obtain permission by contacting the authors. \\
Copyright is held by the authors. This version is prepared for arXiv submission.}
\addtocounter{footnote}{-1}\endgroup

\ifdefempty{\vldbavailabilityurl}{}{
\vspace{.3cm}
\begingroup\small\noindent\raggedright\textbf{PVLDB Artifact Availability:}\\
The source code, data, and/or other artifacts have been made available at \url{\vldbavailabilityurl}.
\endgroup
}
\newpage
\input{sections/1-introduction}
 
\input{sections/2-prelimilaries}

\input{sections/3-problem-analysis}

\input{sections/4-CHK-sequential}

\input{sections/5-CHK-sequential-analysis}

\input{sections/6-CHK-parallel}

\input{sections/7-evaluation}

\input{sections/10-conclusions}

\begin{acks}

Work supported by \grantsponsor{mcdn}{Marie Skłodowska-Curie Doctoral Network}{https://marie-sklodowska-curie-actions.ec.europa.eu/} RELAX-DN, funded by EU under Horizon Europe 2021-2027 FP Grant Agreement nr. \grantnum[www.relax-dn.eu/]{mcdn}{101072456}; \grantsponsor{vr}{Swedish Research Council}{https://www.vr.se/} prj.``EPITOME" \grantnum{vr}{2021-05424}; prj TANDEM (Swedish Energy Agency SESBC, ref. nr. 2021-035871/IEM2022-08); Chalmers AoA Energy-INDEED \& Production-"Scalability, Big Data and AI".
\end{acks}

\balance
\bibliographystyle{ACM-Reference-Format}
\bibliography{sample}

\end{document}

%% file: sections/1-introduction.tex
\section{Introduction}
The problem of finding heavy hitters in a set or stream of items requires identifying the items that appear more times than a specified fraction $\phi$ of the set or the processed stream size~\cite{Misra1982FindingRE}.
Besides, it is often needed by downstream applications also to return the estimated heavy hitters' frequencies.
Corporations such as AT\&T~\cite{cormode2004holistic}, Google~\cite{pike2005interpreting}, and Cloudflare~\cite{rakelimit2020} extensively use heavy-hitter detection solutions for network traffic analysis~\cite{Yang2019HeavyKeeperAA, shi2023cuckoo, Yang2018ElasticSA}, anomaly detection~\cite{lakhina2004characterization}, and iceberg query processing~\cite{beyer1999bottom, fang1999computing}. 
These algorithms are also implemented in 
analytics platforms including Druid, Redis, and Databricks~\cite{druid_topn, redis_top_k, databricks_topk}.
Most recently, heavy-hitter detection has also seen prominent applications in large language models, where identifying frequently accessed tokens can improve inference throughput significantly~\cite{zhang2023h2o}.

Considering the size and input rate of the data, it is important to find algorithms that have favorable memory requirements and system alignment (e.g., cache-friendliness), as well as capabilities to process stream items promptly. 
Knowing that the exact solution requires memory linear to the number of distinct items in the stream~\cite{garofalakis2016data}, and that many applications accept some approximation, a substantial volume of literature focuses on succinct (sublinear) representations from which heavy hitters and their frequencies can be queried approximately. 
A common formulation is the $\epsilon$-$\phi$-\emph{heavy hitters} problem, where the requirement is to return the estimated heavy hitters and their estimated frequencies, approximated within a bounded difference {$\epsilon$} from their true frequencies.
If, furthermore, the requirement allows that the estimated heavy hitter and its frequency are within the bounded difference with a probability of at least $1-\delta$, the problem is known as the $(\epsilon, \delta)$-$\phi$-\emph{heavy hitters}. Such a relaxation can align with an even lower memory footprint.

\textit{Challenges.}  
Given the same amount of \emph{memory}, $\epsilon$-$\phi$-\emph{heavy hitters} algorithms can be compared based on their \emph{throughput and accuracy} across different workloads.
They can be grouped into key-value \emph{(KV)-based}, \emph{sketch-based}, and \emph{hybrid algorithms}.
\emph{KV-based} algorithms such as \algoemph{Frequent}~\cite{Misra1982FindingRE}, \algoemph{LossyCounting}~\cite{Manku2012ApproximateFC}, and \algoemph{Space-Saving}~\cite{Metwally2006AnIE} maintain a fixed number of key-value counters to find the heavy hitters deterministically. 
While these perform well in finding heavy hitters due to their deterministic nature, they suffer from significant frequency estimation errors and throughput challenges. 
\emph{Sketch-based} algorithms such as \algoemph{Count-MinSketch}~\cite{Cormode2004AnID} and \algoemph{CountSketch}~\cite{Charikar2002FindingFI} typically use a series of scalar counters arranged in a two-dimensional array with multiple hash functions to map input items to counters;
however, collisions can cause infrequent items to be mistaken for frequent ones.
\emph{Hybrid} algorithms such as 
\algoemph{HeavyGuardian}~\cite{Yang2018HeavyGuardianSA},
\algoemph{ElasticSketch}~\cite{Yang2018ElasticSA},
\algoemph{AugmentedSketch}~\cite{roy2016augmented},
\algoemph{CuckooCounter}~\cite{shi2023cuckoo},
\algoemph{Topkapi}~\cite{Mandal2018TopkapiPA}, and
\algoemph{HeavyKeeper}~\cite{Yang2019HeavyKeeperAA} combine the strengths of both KV-based and sketch-based approaches to provide better results regarding throughput, memory usage, and accuracy. 
Nevertheless, depending on the ways techniques are combined and used, there are synergies and trade-offs, which call for better balancing of competing requirements and further improvements.\looseness=-1

Moreover, as data volumes and rates continue to grow while applications evolve, there is an increasing need for solutions that go beyond sequential processing capabilities, posing \emph{challenges on both scalable parallel performance and support for concurrent queries and updates}~\cite{garofalakis2016data}. 
While there is some recent work realizing such needs~\cite{Rinberg2020IntermediateVL, Rinberg2022FastCD, Stylianopoulos2020DelegationSA,hilgendorf2025lmqsketchlagommultiquerysketch}, parallelizing $\epsilon$-$\phi$-\emph{heavy hitters} algorithms has focused mainly on parallelizing insertions (cf.~\cite{Zhang2014AnEF, Mandal2018TopkapiPA}), with only one exception~\cite{jarlow2024qpopss}, to the best of our knowledge, on concurrent queries and insertions.

\textit{Contributions.} 
\vinhrevise{}{We provide ways to balance and improve the multi-faceted trade-offs of the problem, in two orthogonal directions:}

First, we propose \emph{\algoname}{} (CHK), a fast, accurate, and space-efficient \emph{$\epsilon$-$\phi$-heavy hitters} algorithm.
CHK {inverts} and {repurposes} components in the conventional data flow seen in hybrid approaches~\cite{Yang2018HeavyGuardianSA, roy2016augmented, Yang2018ElasticSA} --- instead of placing items directly into the \textit{heavy part} and redirecting them to a fallback storage (where estimated frequencies have higher relative error) when needed, CHK requires items to first pass through the \textit{lobby part} and prove their significance before being promoted to the \textit{heavy part}; i.e. the \textit{lobby} acts as a lightweight filter to identify potential heavy hitters, while the \textit{heavy part} maintains accurate counts of heavy candidates through hash collision resolution.
This inversion unlocks new algorithmic synergies  (e.g., applying hash collision resolution selectively to heavy-hitter candidates), which were inaccessible with common existing approaches, thus resulting in better throughput, memory efficiency, and accuracy, even with low-skew data.
This key idea also enables an algorithmic implementation of CHK with a system-aware layout and calculation optimizations.

Second, we introduce a parallel algorithmic framework for heavy-hitter detection that supports concurrent insertions and queries via \textit{domain-splitting}~\cite{Stylianopoulos2020DelegationSA} (partitioning the input and assigning each partition to a thread). 
This framework offers two complementary designs:
\emph{\paralgonamei}, targeting workloads where insertions and frequency queries are predominant compared to heavy-hitter queries, and \emph{\paralgonameq}, for scenarios where heavy-hitter queries are more frequent.
Notably, our parallel designs are compatible with any heavy-hitter detection algorithm without requiring the underlying data structures to support mergeability~\cite{agarwal2013mergeable}. This flexibility allows a wider range of algorithms to be parallelized using our approach.

Moreover, we provide analytical bounds for the algorithms, an open-source repository with their implementations~\cite{ngo_2025_15593109}, as well as a comprehensive experimental evaluation on both real-world and synthetic datasets with diverse distributions, across different hardware platforms and varying query selectivity, in comparison with an array of representative established methods in the literature.
The results show that CHK and its parallel counterparts generalize across varied workloads and offer predictable performance in practical deployments where the stream properties and memory requirements cannot be known in advance. They can process streams of diverse skewness with orders of magnitude improved accuracy, even under memory limitations. 
This improvement directly translates to parallel performance benefits --- with fewer infrequent items incorrectly identified as heavy hitters, threads need less synchronization to maintain consistency, resulting in less data movement overhead.
This implies clear benefits in both timeliness and scalability. 
All these properties make the CHK algorithm family a powerful and potentially influential component in tool-chains for large-scale data analytics.

\begin{table}[t!]
  \footnotesize
  \captionsetup{font=small}
  \captionsetup{skip=2.7pt} 
  \caption{Notation Summary for Heavy Hitters Problem}
  \label{tab:general_notation} 
  \centering
  \begin{tabularx}{\linewidth}{|l|X|} 
  \hline
  \textbf{Notation} &\multicolumn{1}{|c|}{\textbf{Description}} \\
  \hline
  $S$ & Input stream of item tuples $(a_1, \ldots, a_t, \ldots)$ \\
  $N$ & Total weighted size of the processed data stream \\
  $a_t = (e, w)$ & Tuple at timestamp/position $t$, where $e \in U$ is an item drawn from the input universe $U$ with weight $w$ \\
  $\phi$ & Heavy-hitter frequency threshold \\
  $\epsilon$ & Approximation bound for frequency estimation \\
  $\delta$ & Confidence parameter for probabilistic guarantees \\
  $f(e)$ & True frequency of item $e$ \\
  $\hat{f}(e)$ & Estimated frequency of item $e$ \\
  $R$ & Set of true heavy hitters $\{\langle e, f(e)\rangle \mid f(e) \geq \phi N\}$ \\
  $\hat{R}$ & Set of estimated heavy hitters returned  \\
  \hline
  \end{tabularx}
\end{table}

\textit{Roadmap.} 
In \S\ref{sec:prelim} we describe the problem, followed by its analysis relative to related work in \S\ref{sec:prob-analys}. 
In \S\ref{sec:seqCHK} and \S\ref{sec:approximation_bounds} we detail the CHK algorithm design and its analysis, while \S\ref{sec:parallelCHK} is about the parallel algorithm designs, also in association with related work.
In \S\ref{sec:eval} we present our detailed empirical evaluation and we conclude in \S\ref{sec:concl}.

%% file: sections/2-prelimilaries.tex
\section{Problem description}
\label{sec:prelim}

Given a data stream, heavy hitters are items whose frequency exceeds a threshold $\phi$ of the processed stream size $N$.
The problem was first formally described by Misra and Gries~\cite{Misra1982FindingRE} as follows:

\noindent\textbf{$\epsilon$-$\phi$-heavy hitters:} 
Given a stream $S = (a_1, \ldots, a_t, \ldots)$ where each $a_t = (e, w)$ represents a tuple with $e$ being an item from the input universe $U$, $t$ being a timestamp, $w$ (a positive integer) being the weight of the tuple, and  $f(e)$ denoting the true frequency of item~$e$, where each update $a_t = (e, w)$ increments $f(e)$ by $w$. 
Let $N = \sum_{e \in U} f(e)$ denote the total weighted size of the processed stream,  $\phi \in (0, 1)$ denote the heavy-hitter threshold and $\epsilon \in (0, 1)$ ($\epsilon \ll \phi$) denote the approximation bound. 
Let $R = \{\langle e, f(e)\rangle \mid f(e) \geq \phi N\}$ be the set of true heavy hitters in the processed stream and $\hat{R} \subset \{\langle e, \hat{f}(e)\rangle \mid e \in U\}$ be the set of estimated heavy hitters returned by the query. 
The \emph{$\epsilon$-$\phi$-heavy hitters} must satisfy:
\begin{enumerate}[labelwidth=\widthof{\textbf{C3}:},itemindent=!,noitemsep,topsep=0pt]
\item[\textbf{C1}:]\label{cond:no-false-neg} If $f(e) \geq \phi N$, then $e \in \hat{R}$ (no false negatives)
\item[\textbf{C2}:] \label{cond:limited-false-pos} If $f(e) \leq (\phi - \epsilon) N$, then $e \notin \hat{R}$ (limited false positives)
\item[\textbf{C3}:] \label{cond:bounded-approx-dev} For each $e \in \hat{R}$, $\mid f(e) - \hat{f}(e) \mid \leq \epsilon N$ (bounded deviation)
\end{enumerate}

Let $\delta \in (0,1)$ denote the confidence parameter for probabilistic guarantees; an \textbf{$(\epsilon, \delta)$-$\phi$-heavy hitters} algorithm guarantees:

\begin{enumerate}[leftmargin=0pt,labelwidth=\widthof{\textbf{C3}:},itemindent=!,noitemsep,topsep=0pt]
\item[\textbf{C4}:] \label{cond:prob} For each condition (C1–C3), if the premise is met, then with probability at least $1-\delta$, the stated outcome holds.
\end{enumerate}

We primarily focus on the  $(\epsilon, \delta)$-$\phi$-heavy hitters problem, and use deterministic variants as baselines for theoretical and empirical comparisons. The notation summary is provided in Table~\ref{tab:general_notation}.

Heavy-hitter data structures support the following operations:

\noindent\textbf{Update(e, w):} Given an item $e \in U$ and weight $w$, processes the stream tuple $(e,w)$ and maintains necessary data structure state.

\noindent\textbf{f-Query(e):} 
returns $e$'s estimated frequency $\hat{f}(e)$. If $e$ is a heavy hitter ($e \in R$), the estimate satisfies conditions (C1-C4).

\noindent\textbf{hh-Query():} 
Returns the set $\hat{R}$ of estimated heavy hitters along with their estimated frequencies $\hat{f}(.)$. The returned results must satisfy conditions (C1-C4).

\noindent\textbf{Metrics of interest.} 
(1)~\emph{Precision} $\left(\frac{|R \cap \hat{R}|}{|\hat{R}|}\right)$ measures the fraction of reported heavy hitters that are true ones. 
(2)~\emph{Recall} $\left(\frac{|R \cap \hat{R}|}{|R|}\right)$ measures the fraction true heavy hitters identified. 
(3)~\emph{Average Relative Error (ARE)} $\left(\frac{1}{|R|}\sum_{e \in R}\frac{|f(e) - \hat{f}(e)|}{f(e)}\right)$ measures the deviation between true and estimated frequencies across true heavy hitters. 
(4)~\emph{Throughput} measures the number of operations processed per time unit. 
(5)~\emph{Query latency} measures the time to process a query. 
(6)~\emph{Memory usage} measures the space needed for data structures.

%% file: sections/3-problem-analysis.tex
\section{Related Work and Problem Analysis}
\label{sec:prob-analys}

\subsection{Traditional approaches}
\algoemph{Key-Value (KV)-based algorithms} such as \algoemph{Frequent}~\cite{Misra1982FindingRE}, \algoemph{LossyCounting}~\cite{Manku2012ApproximateFC}, and \algoemph{Space-Saving}~\cite{Metwally2006AnIE} maintain a fixed set of approximately $1/\phi$ key-value counters  
to track the frequencies of heavy hitters. 
Because of this fixed size, only potential heavy-hitter candidates can be tracked.
Upon item arrival, if the item is already tracked, its counter is incremented; otherwise, the algorithm must either allocate a new counter (if available) or reassign an existing one. 
The returned heavy hitters set satisfies  \textbf{(C1-C3)} (cf. \S\ref{sec:prelim}). However, their estimated frequencies may suffer from significant errors due to over-/under-estimation, and/or their update operations can be computationally expensive, which can affect downstream tasks~\cite{cormode2008finding}.

\algoemph{Sketch-based algorithms}, such as the \algoemph{CountSketch} and \algoemph{Count-MinSketch}~\cite{Cormode2004AnID, Charikar2002FindingFI}, use a series of scalar counters arranged in a  $d \times w$ array.
Each row corresponds to one of $d$ pairwise-independent hash functions $h_1, h_2, \ldots, h_d$, each $h_i$ mapping items from  $U$ to one of the $w$ counters in row $i$. 
Each input tuple's item is hashed with each $h_i$  to determine which counters to update. 
The estimated frequency of an item is calculated by aggregating the counts from its hashed positions (taking the average or minimum in the aforementioned methods), satisfying conditions \textbf{C1-C4} (cf. \S\ref{cond:no-false-neg}).
However, since multiple items may be mapped to the same counter, low-frequency items may be mistakenly identified as high-frequency ones, leading to incorrect identification of heavy hitters~\cite{cormode2008finding}.

\subsection{Recent advances}
\label{sec:recent-advances}
Based on the strengths of each of the approaches, new techniques provided advances in throughput, memory usage, and accuracy.

\subsubsection{Advances regarding Throughput.} 
Sketch-based approaches achieve better throughput compared to KV-based ones, due to lower time complexity for updates, through hashing.
Later approaches target improved throughput in similar ways by combining hashing with counter-based methods, such as \algoemph{Space-Saving Heap}~\cite{cormode2008finding},  \algoemph{HeavyKeeper}~\cite{Yang2019HeavyKeeperAA}, and \algoemph{Topkapi}~\cite{Mandal2018TopkapiPA}, or by adopting faster hashing schemes like \algoemph{cuckoo hashing}~\cite{Pagh2001CuckooH} in \algoemph{CuckooCounter}~\cite{shi2023cuckoo}.

\subsubsection{Advances regarding Memory Usage.}
In many real-world data streams, only a small fraction of items appear frequently enough to become heavy hitters. 
This suggests dividing the data structure into a \textit{heavy} and a \textit{light part}, each handling either frequent or infrequent items. 
Hence, more suitable data structures with proper sizes, potentially smaller for the less important parts, can be used for each substructure.
Commonly, the \textit{heavy part} is implemented as a simple $\langle Key, Frequency \rangle$
{key-value data structure} to store heavy hitters.
When an item is inserted, the algorithm first checks the \textit{heavy part}; if the item is there, its count is updated; else, the item is inserted into a sketch which acts as the \textit{light part}.
\algoemph{HeavyGuardian}~\cite{Yang2018HeavyGuardianSA}, \algoemph{ElasticSketch}~\cite{Yang2018ElasticSA}, and \algoemph{AugmentedSketch}~\cite{roy2016augmented} adopt this approach.

\subsubsection{Advances regarding Accuracy.}
To improve the accuracy of the estimated frequency of heavy hitters, the \textit{count-with-exponential-decay} method was introduced and used in \algoemph{HeavyGuardian}~\cite{Yang2018HeavyGuardianSA} and \algoemph{HeavyKeeper}~\cite{Yang2019HeavyKeeperAA}.
This technique modifies how counter values are updated when hash collisions occur.
Unlike the key-value counter in the Frequent algorithm~\cite{Misra1982FindingRE} that uniformly decrements counters regardless of their values, the \textit{count-with-exponential-decay} lowers the counter by 1 with probability \( b^{-C} \), for every unit of the update weight, where \( b \) is a decay base and \( C \) is the current counter value.
If the item is a heavy hitter, $C$ is high, hence, the probability to decrement the counter will be lower, and the item retains a more accurate frequency value.

\subsubsection{Synergetic effects and Trade-offs}
The aforementioned techniques are different high-level concepts with a variety of algorithmic design choices and combinations.
Each choice or combination means different synergies or trade-offs.
For example, frequent/infrequent separation not only reduces the memory used but potentially improves the throughput through a fast path created by the part storing frequent items, and also has better cache behavior, as seen in \algoemph{HeavyGuardian}~\cite{Yang2018HeavyGuardianSA} and \algoemph{ElasticSketch}~\cite{Yang2018ElasticSA}.
However, these techniques can show opposite effects under different conditions. In lower-skew datasets, \algoemph{HeavyGuardian} and \algoemph{ElasticSketch} may suffer accuracy loss due to frequent collisions, and \algoemph{AugmentedSketch}~\cite{roy2016augmented} may experience reduced throughput from frequent data movement between parts.
Moreover, applying \textit{count-with-exponential-decay} on all counters regardless of the item's status can hurt throughput due to floating point operations.

%% file: sections/4-CHK-sequential.tex
\section{Sequential \algoname{}}
\label{sec:seqCHK}
By studying the strengths and multi-faceted trade-offs of different designs, we aim to balance and improve the metrics of interest by combining key algorithmic elements in a novel, system-aware fashion, that harmonizes their benefits
and enables suitability for varying input features. 
We propose a new algorithm, \emph{\algoname{}}, which represents a fundamental rethinking of the frequent/infrequent separation concept for heavy-hitter detection. 
This reinterpretation --- as will be clarified in the following --- unlocks synergies that were previously inaccessible with conventional approaches. 
We start by outlining the data structure layout, followed by a discussion of the design choices. Next, we explain the algorithm's operations, the weighted update, and other optimizations.\looseness=-1

\subsection{\algoname{} data structure layout}
Fig.~\ref{fig:cuckoo_heavy_keeper_structure} shows {the layout of} the \emph{\algoname}. 
The notation related to the data structure is summarized in Table~\ref{tab:CHK_pseudocode_notation}.
\emph{\algoname} maintains two tables $T[0]$ and $T[1]$\footnote{\algoname{} can be implemented using two separate tables or a single one with two hash functions; we use the two-table approach (as in the original cuckoo hashing work~\cite{Pagh2001CuckooH}), which guarantees distinct candidate locations.}, each consisting of an array of buckets.
Each bucket has one lobby entry for filtering infrequent items and multiple heavy entries (e.g., 2-4 per bucket, cf. \S\ref{sec:eval} for parameter selection rationale) for maintaining heavy-hitter candidates.
Each bucket's lobby entry ($T[i][idx_i].lobby$) stores a tuple $\langle {\mathit{fp}}, C \rangle$ where $\mathit{fp}$ is a fixed-size fingerprint
and $C$ is small counter implementing the \textit{count-with-exponential-decay} method.
The heavy entries ($T[i][idx_i].heavy$) use the same tuple format but with larger counters, for more precise tracking.

\begin{table}[t]
    \footnotesize
    \captionsetup{font=small, skip=2.7pt}
    \caption{Notation Summary for \algoname{}}
    \centering
    \begin{tabularx}{\linewidth}{|l|X|}
        \hline
        \textbf{Notation} & \multicolumn{1}{|c|}{\textbf{Description}}  \\
        \hline
        $T[2][]$ & Two tables of buckets \\
        $\mathit{fp}$ & Fingerprint of item $e$ \\
        $idx_0$, $idx_1$ & Indices in the hash tables computed from $e$ \\
        $T[i][idx_i]$ & Bucket at position $idx_i$ in table $i$ ($i \in \{0,1\}$) \\
        $T[i][idx_i].lobby$ & Lobby entry in bucket $T[i][idx_i]$ \\
        $T[i][idx_i].heavy$ & Array of heavy entries in bucket $T[i][idx_i]$ \\
        $\Thp$ & Promotion threshold for lobby entries \\
        $\mathcal{B}$ & Number of buckets in each table \\
        $b$ & Decay base used in counter decay \\
        $de[k]$ & Expected decays to reduce $C=k$ to 0 \\
        $\texttt{MAX\_KICKS}$ & Maximum number of kicks allowed \\
        \hline
    \end{tabularx}
    \label{tab:CHK_pseudocode_notation} 
\end{table}
\subsection{\algoname{} key ideas}

\textbf{Count-with-Exponential-Decay as a filter.} 
We observe that \textit{count-with-exponential-decay}, typically used for counting heavy items (in \textit{heavy part}) in prior work, can be repurposed as a \textit{lobby} that holds only \emph{potential} heavy hitters, as it allows frequent items to accumulate count while infrequent items fade quickly. 
Our approach \emph{inverts} the conventional workflow in frequent/infrequent separation: while algorithms~\cite{Yang2018HeavyGuardianSA, roy2016augmented, Yang2018ElasticSA} direct items first to the \textit{heavy part} and use a \textit{light part} as fallback storage when no space remains or none of the existing items in the \textit{heavy part} can be replaced, our \emph{\algoname} eliminates the \textit{light part} entirely. 
Instead, items must first go through a \textit{lobby} that filters potential heavy hitters eligible for promotion to the \textit{heavy part}. 
Once an item is identified as a heavy-hitter candidate, it is moved to the \textit{heavy part} where it is tracked accurately without decay operations.

\textbf{Collision resolution in the \textit{heavy part}.} 
Most \textit{heavy part} implementations use simple key-value data structures without hash collision handling ~\cite{Yang2018HeavyGuardianSA, roy2016augmented, Yang2018ElasticSA}, meaning that, for example, in low-memory settings, heavy items can be hashed to the same bucket, forcing one to be moved to the \textit{light part} where it is tracked less accurately. 
This motivates our use of \algoemph{cuckoo hashing}~\cite{Pagh2001CuckooH} inside the \textit{heavy part} to give colliding heavy hitters a second chance,
which significantly improves recall in diverse settings, particularly under memory constraints and low-skew distributions where collisions among true heavy hitters are more likely to occur. 
Importantly, we only resolve hash collisions for heavy-hitter candidates, which is only possible due to our repurposed \textit{lobby}. 
In contrast, algorithms like \textit{CuckooCounter}~\cite{shi2023cuckoo} resolve collisions for all items, which implies trade-offs in throughput depending on stream cardinality.
\noindent 
\begin{figure} 
    \centering 
    \footnotesize 
    \setlength{\tabcolsep}{1pt}
    \renewcommand{\arraystretch}{0.8}
    \begin{tikzpicture}[
      baseline, 
      class/.style={ 
        rectangle split, 
        rectangle split parts=2, 
        draw, 
        minimum width=1.9cm,
        inner sep=2pt,
        outer sep=0pt,
        rectangle split part fill={white,white},
        align=center, 
        rectangle split part align={center,left}
      }, 
      composition/.style={ 
        draw, 
        diamond, 
        fill=black, 
        minimum size=0.1cm, 
        inner sep=0pt,
        outer sep=0pt
      }
    ]
    
    \node[class] (keeper) {CHK \nodepart{second} tables: Table[2]};
    \node[class, right=0.6cm of keeper] (table) {Table \nodepart{second} buckets: Bucket[]};
    \node[class, right=0.6cm of table] (bucket) {Bucket \nodepart{second}\begin{tabular}[c]{@{}l@{}}lobby: LobbyEntry\\heavy: HeavyEntry[]\end{tabular}};
    \node[class, below=0.3cm of bucket] (heavy) {HeavyEntry \nodepart{second}\begin{tabular}{l}fp: uint16\\C: uint32\end{tabular}};
    \node[class, left=0.3cm of heavy] (lobby) {LobbyEntry \nodepart{second}\begin{tabular}{l}fp: uint16\\C: uint8\end{tabular}};
    
    \draw[dashed] (keeper.east) node[composition] {} -- (table.west);
    \draw[dashed] (table.east) node[composition] {} -- (bucket.west);
    \draw[dashed] (bucket.south) node[composition] {} -- (lobby.north);
    \draw[dashed] (bucket.south) node[composition] {} -- (heavy.north);
\end{tikzpicture}
\captionsetup{font=small}
\caption{
    \emph{\algoname} consists of two tables of buckets. Each bucket has one lobby entry to filter infrequent items and multiple heavy entries to maintain heavy-hitter candidates.
    } 
\label{fig:cuckoo_heavy_keeper_structure}
\Description{Class diagram showing the structure of CHK with classes Table, Bucket, LobbyEntry, and HeavyEntry. The diagram shows their relationships and attributes using UML-style notation.}
\end{figure}

\textbf{System-Aware Design.}
Most existing works focus primarily on asymptotic complexity, which cannot capture system-awareness aspects such as cache efficiency, data movement patterns, and computational costs of floating-point arithmetic. 
These factors, however, significantly influence algorithms' running time (cf.~\cite{drepper2007every} and references therein). 
While maintaining good asymptotic complexity, we also recognize the importance of system-level considerations in our design.
The bucket layout illustrated in Fig.~\ref{fig:cuckoo_heavy_keeper_structure} considers these system-level factors and offers two advantages.
First, storing the lobby entry and heavy entries together in the same bucket allows the bucket indices to be calculated only once, for both the \textit{lobby} and the \textit{heavy part}. 
Second, when accessing any entry in a bucket, the CPU cache line brings in all the entries in that bucket, which helps to check and update the frequencies of items in both parts efficiently, minimizing memory bandwidth contention. 
The approach combines the best design practices of \emph{cuckoo hashing}, which was shown analytically to improve cache utilization and increase occupancy rates~\cite{Dietzfelbinger2005BucketizedCuckoo,Fountoulakis2011CuckooProof}.
Furthermore, our design restricts the need for floating point operations of \textit{count-with-exponential-decay} within a small range, bounded by the threshold of the \textit{lobby part} (as large counters ``live" only in the \textit{heavy part}, where they are no longer subject to floating point operations). 
This enables the possibility of pre-calculation and tabulation, to replace expensive operations with simple lookups, as we elaborate in subsequent sections.

\subsection{\algoname{}  main operations}
\label{sec:chk_main_operations}
For an item $e$, the algorithm stores a fixed-size fingerprint $\mathit{fp} = fingerprint(e)$; its two possible mapped bucket indices are calculated as $idx_0 = hash(e)$ and $idx_1 = idx_0 \oplus hash(\mathit{fp})$. 
Given either $idx_0$ or $idx_1$ and the fingerprint, the other index can be derived using $XOR$ operations.
This technique, known as \emph{partial-key cuckoo hashing}~\cite{Fan2014CuckooFP}, reduces memory footprint by using smaller fingerprints while maintaining a low false positive rate for identity checks.
\input{sections/pseudo-code/CHK-main-operations.tex}

\input{sections/pseudo-code/CHK-helper-functions.tex}
\noindent\textbf{Update} (Alg.~\ref{alg:chk-main}):
For each input tuple $(e, w)$, the fingerprint $\mathit{fp}$ and both bucket indices $idx_0, idx_1$ for $T[0], T[1]$ are generated using independent hash functions (Alg.~\ref{alg:chk-helpers}, \hyperref[alg:chk-main]{\GenerateFpAndIndexes} function, cf.~\cite{Fan2014CuckooFP}).
The update process follows three cases:

\noindent\textbf{Case 1 - Item is tracked in \textit{heavy part}} (Alg.~\ref{alg:chk-main}, l. 4-6):  
If $\mathit{fp}$ matches an entry in $T[i][idx_i].heavy$ of either possible bucket, the algorithm increments the matched entry's counter by $w$ and returns.
This is the \emph{most common and fastest path} (green part in pseudo-code)  --- since tracked heavy-hitter candidates account for a large portion of the stream, this path will be taken most of the time.
Additionally, since all heavy entries in a bucket are in the same cache line, they can be checked and updated without additional memory accesses.

\noindent\textbf{Case 2 - Item is tracked in lobby part} (Alg.~\ref{alg:chk-main}, l. 7-9): 
If $\mathit{fp}$ matches a lobby entry, the entry's counter $C$ is incremented; if it exceeds $\Thp$, a threshold parameter, {the algorithm attempts} promotion (Alg.~\ref{alg:chk-helpers}, \hyperref[fn:try-promote]{\TryPromote} function) for the item from \textit{lobby} to \textit{heavy part}: 
    \begin{itemize}[leftmargin=*]
        \item If an empty heavy entry exists in the bucket, the item is promoted directly to the \textit{heavy part}, along with its counter. 
        \item Otherwise, {it tries} promotion, which succeeds with probability $(C - L)/(C_{min}^{heavy} - L)$ where $C_{min}^{heavy}$ is the smallest counter in the bucket's heavy entries (Fig.~\ref{fig:chk_trypromote_kickout}).
        If the promotion succeeds, the promoted item's counter becomes $C_{min}^{heavy}$ and cuckoo kickout is initiated, to relocate displaced entries (Alg.~\ref{alg:chk-helpers}, \hyperref[fn:kickout]{\Kickout} function).
        The \hyperref[fn:kickout]{\Kickout} function relocates the displaced entry to its alternate bucket ("\texttt{alt bucket}" in Alg.~\ref{alg:chk-helpers}, l.44) in the other table. If the alternate bucket has an empty heavy entry, the displaced entry moves there directly. Otherwise, it displaces the heavy entry with the minimum counter, which is then recursively kicked out. This recursive process continues until an empty heavy entry is found, or \texttt{MAX\_KICKS} attempts are reached, or the counter of a displaced item falls below $\phi N$\footnote{This threshold for viable heavy-hitter candidates can be adjusted. For example, setting it to $0$ disables the early cuckoo kickout termination.}, indicating it is no longer a viable heavy-hitter candidate.
        If the promotion fails, the item remains in the \textit{lobby part}, with its counter set to $\Thp$. 
        This probabilistic promotion ensures that higher-frequency items are more likely to be promoted to the \textit{heavy part}.
    \end{itemize}

\begin{figure*}
    \centering
    \includegraphics[width=\textwidth]{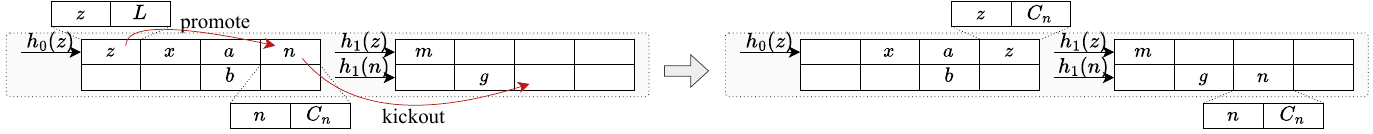}
    \captionsetup{justification=centering, font=small}
    \caption{lobby item $z$ replaces the minimum-counter item $n$ in the heavy part, which then triggers relocation via cuckoo hashing. The process may continue recursively until finding an empty slot, or reaching \texttt{MAX\_KICKS}, or encountering a below-threshold item.}
    \label{fig:chk_trypromote_kickout}
    \Description{A diagram showing the cuckoo kickout process: item z being promoted from lobby to heavy part, replacing item n, which then attempts to find space in another table through cuckoo hashing.}
\end{figure*}

\noindent\textbf{Case 3 - Item is not tracked} (Alg.~\ref{alg:chk-main}, l. 10-22): 
If there is an empty lobby entry in either of the buckets that the item is hashed to, the algorithm inserts $\langle \mathit{fp}, w \rangle$ directly, promotes it if $w \geq L$, and returns. Otherwise, it applies \textit{count-with-exponential-decay} to the target lobby entry, which is determined using \textit{modulo divsion hashing} ($\mathit{fp}\bmod2$) to ensure consistent bucket selection.
For unweighted updates ($w = 1$), the procedure decays the counter with probability $b^{-C}$, where $b$ is the decay base and $C$ is the current counter value. After the decay operation, if the counter $C=0$, the existing fingerprint in $T[i][idx_i].lobby$ is replaced with the fingerprint $\mathit{fp}$ of the incoming item $e$. 
For weighted updates ($w > 1$), the \hyperref[fn:decay]{\DecayCounter} function in Alg.~\ref{alg:chk-helpers} simulates a sequence of unweighted updates: it calculates how many decay operations would be needed for the existing counter to reach zero, and then compares this with the incoming weight. 
Based on this comparison, the algorithm either replaces the lobby item when the weight can fully decay the counter (with promotion if the remaining weight exceeds $\Thp$), or partially decrements the existing counter based on statistical projection. 
This method ensures that those with large weights can quickly establish themselves as heavy-hitter candidates. 
The analysis of weighted update behavior is presented in \S\ref{sec:weighted_update}. 
For timeliness and practical efficiency, \hyperref[alg:chk-helpers]{\DecayCounter}, in this context can be realized through tabulation, as it applies to bounded counting, making this path be very fast as well. 
This optimization is detailed in \S\ref{sec:optimizations}.

\noindent\textbf{f-Query($e$):} (Alg.~\ref{alg:chk-main}, \hyperref[alg:chk-main]{\Query} function) When querying the estimated frequency $\hat{f}(e)$ for a specific item $e$, the algorithm first computes its fingerprint $\mathit{fp}$ and bucket indices $idx_0$, $idx_1$ as in update operations.
Then, it checks all heavy entries in the corresponding buckets for a matching fingerprint.
If found, the counter value is returned; otherwise, the item is not tracked, and the query returns~0.

\noindent\textbf{hh-Query():}
The implementation can vary by the specific use case (e.g., offline or streaming environment, if queries are continuous). For this reason, we omit the pseudo-code. 
However, for completeness, we describe a common approach using an auxiliary min-heap data structure that maintains the set of heavy hitters $\hat{R}$ continuously during stream processing.
During \hyperref[alg:chk-main]{\Update}, if $\hat{f}(e) \geq \phi N$, the item is pushed or updated in the heap. 
The algorithm then repeatedly removes the root item if its frequency falls below $\phi N$ until all remaining items exceed the threshold. 
When querying heavy hitters, all items in the heap are returned with their frequencies.

\subsection{Weighted update} 
\label{sec:weighted_update} 
Most ($\epsilon, \phi$)-heavy hitters algorithms focus on unweighted updates ($w=1$).
\footnotetext[3]{Our microbenchmarks~\cite{ngo_2025_15593109} of the weighted upate, following \hyperlink{box:param-config}{\textbf{Recommended Parameter Configuration}} in \S\ref{sec:eval} ($L=16$, $b=1.08$) show that tabulation is 168$\times$ faster than repeated unweighted updates and 2.6$\times$ faster than theorem~\ref{thm.weighted_decay} formula calculation.}
However, there can be a significant benefit in the ability to handle weighted updates ($w>1$),
e.g., process aggregated data from upstream tasks in distributed processing pipelines, which can substantially reduce communication costs between system components.
When handling weighted updates, many algorithms perform multiple unweighted updates, which can degrade performance or produce incorrect results, especially in concurrent settings  (cf. \S\ref{sec:parallelCHK}) where thread interference may occur during repeated updates.
To address this limitation, we develop a weighted update extension for the \textit{count-with-exponential-decay} that we incorporate into CHK. Note that this extension can be directly adopted by other systems using similar counting techniques (e.g., Redis Top-k \cite{redis_top_k}).

Let $dc_{C,w}$ denote the counter value after applying $w$ consecutive decay operations to an initial counter value $C$.
Each single decay operation ($w=1$) sets $C=C-1$ with probability $b^{-C}$ or keeps it unchanged with probability $1-b^{-C}$, where $b$ is the decay base.
For weighted updates, we need to determine $E[dc_{C,w}]$, the expected counter value after $w$ decay operations.
\begin{theorem}\label{thm.weighted_decay}
    For a counter value $C$ with decay base $b$, where each decay operation sets the value $C=C-1$ with probability $b^{-C}$, the expected counter value $E[dc_{C,w}]$ after $w$  decay operations is:
    \begin{equation}
    E[dc_{C,w}] = \log_b \left(b^C - \frac{w(b - 1)}{b}\right)
    \label{eq:expected_C}
    \end{equation}  
\end{theorem}
\begin{proof}   
    Let $de[k]$ represent the expected number of decay operations needed to decrease a counter from $k$ to 0. Since each decay has probability $b^{-C}$, one successful decay requires an expected $b^C$ attempts. For $C > 0$, using geometric series sum: $de[C] = \sum_{k=1}^C b^k = \frac{b(b^C - 1)}{b - 1}$. Since each decay operation reduces the expected attempts by 1, after $w$ operations: $de[E[dc_{C,w}]] = de[C] - w$. Substituting into both sides and simplifying, we get the result.
\end{proof}

\subsection{Optimizations}
\label{sec:optimizations}
\emph{Precomputed Decay Outcomes for Weighted Updates.}
Computing $E[dc_{C,w}]$ (cf. \S\ref{sec:weighted_update}) for each update involves expensive floating-point operations. 
Recall that our design restricts these calculations to lobby counters with bounded values ($\Thp$).
This enables the possibility of pre-calculation and tabulation to replace expensive operations with simple lookups.
We can precompute in the $de[]$ array the expected number of decay operations needed to reduce a counter from $k$ to~$0$: $de[0] = 0$, and for $k = 1$ to $\Thp$, $de[k] = de[k-1] + b^k$.
When performing an update with weight $w$, if $w \geq de[C]$, that means the incoming weight is larger than the expected number of decay operations needed to reduce the existing counter $C$ to~$0$.
In this case, the incoming item's fingerprint $\mathit{fp}$ replaces the old one, the remaining weight is calculated as $w' = w - de[C]$, and \hyperref[fn:try-promote]{\TryPromote} will be triggered if $w' > \Thp$.
Otherwise, binary search can be used to find the expected value after $w$ decay operations, thus updating $C$ with $C_{new}$.
The pseudo-code for weighted and unweighted updates is highlighted in the \hyperref[fn:decay]{\DecayCounter} function (Alg.~\ref{alg:chk-helpers}).
This optimization trades $O(\Thp)$ memory for $O(\log \Thp)$ lookup time, eliminating all floating-point operations during updates while preserving the same guarantees\footnotemark{}.

\emph{Early heavy-entry placement.}
During the initial stages with empty heavy entries, forcing items through the lobby reduces accuracy due to exponential decay counting. Since heavy hitters often appear early in streams, we introduce early heavy entry placement. When an item arrives and finds an empty heavy entry, it is placed there directly. This allows precise counting of heavy hitters early on.

%% file: sections/pseudo-code/CHK-main-operations.tex
\setlength{\algomargin}{0.2em}
\begin{algorithm}
  \caption{{\algoname} - Main Operations}
  \label{alg:chk-main}
  \footnotesize
  \setstretch{0.8} 
  \SetKwProg{Fn}{Procedure}{}{}
  \SetKwFunction{Update}{Update}
  \SetKwFunction{Query}{f-Query}
  \SetKwFunction{GenerateFpAndIndexes}{GenerateFpAndIndexes}
  \SetKwFunction{CheckAndUpdateHeavy}{CheckAndUpdateHeavy}
  \SetKwFunction{CheckAndUpdateLobby}{CheckAndUpdateLobby}
  \SetKwFunction{DecayCounter}{DecayCounter}
  \SetKwFunction{TryPromote}{TryPromote}

  \Fn{\Update{$e, w$}}{{\label{fn:update}}
      $N \gets N + w$
      
      $\mathit{fp}, idx_0, idx_1 \gets$ \hyperref[fn:generate]{\GenerateFpAndIndexes{$e$}}\;
      
      \boxit{MaterialGreen50}{2.1}
      \tcp{Found and updated in heavy \textcolor{red}{$\triangleright$ most common and fastest path}}
      \If{\hyperref[fn:check-heavy]{\CheckAndUpdateHeavy{$\mathit{fp}, idx_0, idx_1, w$}}}{
          \Return $\hat{f}(e)$\;
          }
          
      \boxit{MaterialBlue50}{10.5}
      \tcp{Found and updated in lobby}
      \If{\hyperref[fn:check-lobby]{\CheckAndUpdateLobby{$\mathit{fp}, idx_0, idx_1, w$}}}{
          \Return $\hat{f}(e)$
      }

    \tcp{If empty lobby exists}
      \If{exists $i \in \{0,1\}$ such that $T[i][idx_i].lobby$ is empty}{
          Insert $\langle \mathit{fp}, w \rangle$ into empty $T[i][idx_i].lobby$\;
          \If{w $\geq \Thp$}{\hyperref[fn:try-promote]{\TryPromote{$T[i][idx_i]$}}}
          \Return $\hat{f}(e)$
      }
      
      \boxit{MaterialYellow50}{11.5}
      \tcp{Count with exponential decay}
          $i \gets \mathit{fp} \bmod 2$\;
          $C_{new} \gets$ \hyperref[fn:decay]{\DecayCounter{$T[i][idx_i].lobby.C, w$}}\;
          \eIf{$C_{new} = 0$}{
              Update $T[i][idx_i].lobby$ with $\langle \mathit{fp}, w - de[C] \rangle$\;
          }{
              $T[i][idx_i].lobby.C \gets C_{new}$
              
          }
          
          \If{$T[i][idx_i].lobby.C \geq \Thp$}{
              \hyperref[fn:try-promote]{\TryPromote{$T[i][idx_i]$}}
          }%
      \Return $\hat{f}(e)$\;
  }
  
  \Fn{\Query{$e$}}{{\label{fn:query}}
      $\mathit{fp}, idx_0, idx_1 \gets$ \GenerateFpAndIndexes{$e$}\;
      \Return $\hat{f}(e)$ if found, 0 otherwise\;
  }

\end{algorithm}

%% file: sections/pseudo-code/CHK-helper-functions.tex
\begin{algorithm}
    \caption{{\algoname} - Helper Functions}
    \label{alg:chk-helpers}
    \footnotesize
    \setstretch{0.8} 
    \SetKwProg{Fn}{Procedure}{}{}
    \SetKwFunction{GenerateFpAndIndexes}{GenerateFpAndIndexes}
    \SetKwFunction{CheckAndUpdateHeavy}{CheckAndUpdateHeavy}
    \SetKwFunction{CheckAndUpdateLobby}{CheckAndUpdateLobby}
    \SetKwFunction{DecayCounter}{DecayCounter}
    \SetKwFunction{TryPromote}{TryPromote}
    \SetKwFunction{Kickout}{Kickout}

    \Fn{\GenerateFpAndIndexes{$e$}}{{\label{fn:generate}}
        $\mathit{fp} \gets fingerprint(e)$, $idx_0 \gets hash(e) \bmod \mathcal{B}$\;
        $idx_1 \gets (hash(\mathit{fp}) \oplus idx_0) \bmod \mathcal{B}$\;
        \Return $\mathit{fp}, idx_0, idx_1$\;
    }

    \Fn{\CheckAndUpdateHeavy{$\mathit{fp}, idx_0, idx_1, w$}}{{\label{fn:check-heavy}}
        Search heavy entries in both tables for matching $\mathit{fp}$\;
        \If{found matching entry or empty slot exists}{
            Update counter or insert entry
            \Return true\;
        }
        \Return false\;
    }
    
    \Fn{\CheckAndUpdateLobby{$\mathit{fp}, idx_0, idx_1, w$}}{{\label{fn:check-lobby}}
        Search lobby entries in both tables for matching $\mathit{fp}$\;
        \If{found}{
            Update lobby counter\;
            \If{counter $\geq \Thp$}{\hyperref[fn:try-promote]{\TryPromote{$T[i][idx_i]$}}}
            \Return true\;
        }
        \Return false\;
    }

    \Fn{\DecayCounter{$C, w$}}{{\label{fn:decay}}
    \boxit{MaterialYellow50}{14.3}
    \tcp{Decay with exponential probability}
    \If{$w = 1$}{
        $prob \gets b^{-C}$\;
        \Return ($Random(0,1) < prob$) ? $C - 1$ : $C$\;
    }
    
    \tcp{Weighted update}
    \If{$w > 1$ and $w < min\_decay$}{
        $min\_decay \gets de[C] - de[C-1]$\;
        $prob \gets w / min\_decay$\;
        \Return ($Random(0,1) < prob$) ? $C - 1$ : $C$\;
    }
    \If{$w \geq de[C]$}{\Return 0\;}
    \Return largest $i$ where $de[i] + w \geq de[C]$\;
    }

    \Fn{\TryPromote{bucket}}{{\label{fn:try-promote}}
    \If{empty slot exists in $bucket.heavy$}{
        Move $bucket.lobby$ to empty slot
        \Return\;
    }
    
    $min \gets$ smallest entry in $bucket.heavy$\;
        \If{$Random(0,1) < \frac{bucket.lobby.C - \Thp}{min.C - \Thp}$}{
            $evicted\_item \gets min$\;
            $min.\mathit{fp} \gets bucket.lobby.\mathit{fp}$\;
            Clear $bucket.lobby$ and \hyperref[fn:kickout]{\Kickout{$evicted\_item$}}\;
        }\Else{
            $bucket.lobby.C \gets \Thp$\;
        }
    }

    \Fn{\Kickout{entry}}{{\label{fn:kickout}}
    \For{$kicks \gets 1$ \KwTo $\texttt{MAX\_KICKS}$}{
        \If{$entry.C < \phi N$}{\Return}
        \If{empty heavy entry in alt bucket}{
            Move $entry$ to empty heavy entry
            \Return\;
        }\Else{
            Swap $entry$ with min entry in alt bucket\;
        }
    }
}
\end{algorithm}

%% file: sections/5-CHK-sequential-analysis.tex
\section{Analysis of \algoname{}}
\label{sec:approximation_bounds}

\begin{lemma}[Heavy Hitter Promotion Guarantee]
  \label{thm:detection_bounds}
  Let $\mathcal{B}$ be the number of buckets in the hash table, $N$ be the stream size, and $e$ be any true heavy hitter $R = \{e \mid f(e) \geq \phi N\}$. Assuming no fingerprint collisions occur and the heavy part has sufficient space to accommodate promoted items without losing true heavy hitters during relocations, the probability that $e$ will be promoted to the heavy part is bounded as:
  \[
  \Pr[e \text{ is promoted to heavy part} \mid f(e) \geq \phi N] \geq 1 - \frac{1}{\phi \mathcal{B}}
  \]
\end{lemma}

\begin{proof}
  For a heavy hitter $e$ with frequency $f(e) \geq \phi N$, we analyze the worst-case scenario: when $e$ arrives, its lobby position is already occupied by another item with a counter value at $\Thp$. 
  Due to the exponential decay with base $b$, approximately $b^{\Thp}$ collisions are required to cause a single decrement. 
  In practice, $e$ only collides with a subset of items hashed to the same bucket, but for our analysis, we conservatively assume it must collide with all other items hashed there. If $f(e)$ exceeds this worst-case collision threshold, then $e$ will eventually be promoted. 
  
  Let an indicator variable $I_{e',j} = 1$ if $h(e') = j$ and~$0$  otherwise, where $h$ maps items to $\mathcal{B}$ buckets. Define $Y_j = \sum_{e' \in U} I_{e',j} \cdot f(e')$, which represents the total frequency of all items hashed to bucket $j$. Then: $E[Y_j] = \sum_{e' \in U} \frac{f(e')}{\mathcal{B}} = \frac{N}{\mathcal{B}}$. For $e$ to eventually be promoted, the key inequality is: $f(e) \geq Z \cdot \left(Y_{h(e)} - f(e)\right) + \Thp$ \label{eq:promotion_condition} where $Z$ represents the average number of collisions needed to cause a single decrement. Since $f(e) \gg \Thp$, we can rearrange to get $(Z+1)f(e) \geq Z \cdot Y_{h(e)} + \Thp$, which means promotion occurs when $Y_{h(e)} \lesssim \frac{Z+1}{Z}f(e)$. 
  
  Therefore, the probability that $e$ is not promoted is bounded by Markov's inequality using $E[Y_{h(e)}] = \frac{N}{\mathcal{B}}$:
  \begin{align}
  \Pr\left[Y_{h(e)} \geq \frac{1 + Z}{Z}f(e)\right] \leq \frac{\frac{N}{\mathcal{B}}}{\frac{1 + Z}{Z} \cdot \phi N} = \frac{1}{\mathcal{B}} \cdot \frac{Z}{(1 + Z) \cdot \phi} \leq \frac{1}{\phi \mathcal{B}} \nonumber
  \end{align}
  where the last inequality follows from $\frac{Z}{1+Z} < 1$ for any $Z > 0$. Thus, $\Pr[e \text{ is promoted to heavy part} \mid f(e) \geq \phi N] \geq 1 - \frac{1}{\phi \mathcal{B}}$.
\end{proof}

\begin{theorem}[Approximation Bounds]
  \label{thm:approximation_bounds}
  Let $\mathcal{B}$ be the number of buckets in the hash table, $N$ be the stream size, and $e$ be any item in the heavy part (including heavy hitters that appear there with high probability as shown in lemma \ref{thm:detection_bounds}). Assuming no fingerprint collisions occur and the heavy part has sufficient space to accommodate promoted items without losing true heavy hitters during relocations, for any positive $\epsilon < 1$, the probability of estimation error exceeding $\epsilon N$ is bounded as:
  \[
  \Pr[|\hat{f}(e) - f(e)| \geq \epsilon N] \leq \frac{1}{\epsilon \mathcal{B}}
  \]
\end{theorem}

\begin{proof}
The true frequency $f(e)$ can be decomposed as $f(e) = f^d(e) + f^p(e) + f^h(e)$, where $f^d(e)$ counts occurrences during the \emph{count-with-exponential-decay phase}, $f^p(e)$ counts occurrences during probabilistic promotion, and $f^h(e)$ counts occurrences after entering the heavy part. The algorithm estimates $\hat{f}(e) = f^d(e) - X^d(e) + m + f^h(e)$, where $X^d(e)$ represents decrements during the decay phase and $m$ is the frequency added upon promotion.

For the underestimation bound:
\begin{align} 
&&& \Pr[\hat{f}(e) \leq f(e) - \epsilon N] = \Pr[-X^d(e) + m \leq f^p(e) - \epsilon N] \nonumber \\
&=&&\Pr[X^d(e) + f^p(e) \geq \epsilon N + m] \leq Pr[X^d(e) + f^p(e) \geq \epsilon N] \nonumber \\
&\leq&& \frac{E[X^d(e) + f^p(e)]}{\epsilon N} \text{\quad(by Markov's inequality)}\label{eq:prob_bound}
\end{align}

Using $Y_{h(e)}$ as defined above, since $X^d(e) + f^p(e) \leq f^d(e) + f^p(e) = f(e) - f^h(e) \leq f(e) \leq Y_{h(e)}$, we have: $E[X^d(e) + f^p(e)] \leq E[Y_{h(e)}] = \frac{N}{\mathcal{B}}$.
Substituting into \eqref{eq:prob_bound}:
\begin{equation}
\Pr[\hat{f}(e) \leq f(e) - \epsilon N] \leq \frac{1}{\epsilon \mathcal{B}}
\label{eq:underestimate_bound}
\end{equation}

For the overestimation bound:
\begin{align} 
&&& \Pr[\hat{f}(e) \geq f(e) + \epsilon N] = \Pr[-X^d(e) + m - f^p(e) \geq \epsilon N] \nonumber \\
&\leq&& \Pr[m - f^p(e) \geq \epsilon N] \leq \Pr[m \geq \epsilon N] \leq \Pr[Y_{h(e)} \geq \epsilon N] \nonumber \\
&\leq&& \frac{E[Y_{h(e)}]}{\epsilon N} = \frac{N/\mathcal{B}}{\epsilon N} = \frac{1}{\epsilon \mathcal{B}} \text{\quad(by Markov's inequality)} \label{eq:overestimate_bound}
\end{align}

By the union bound on \eqref{eq:underestimate_bound} and \eqref{eq:overestimate_bound}, the theorem follows.
\end{proof}

From this theorem, we can establish that CHK satisfies the probabilistic heavy-hitter conditions (\textbf{C1-C4}) (cf. \S\ref{cond:no-false-neg}) with appropriate parameterization.
A key assumption for our algorithm—that "the heavy part has sufficient space to accommodate promoted items"—represents, in fact, a significant constraint for other algorithms in practice 
(cf. \S\ref{sec:eval_seq_methodology} for more detail). 
It is also worth noting that our analysis's conservative bounds do not account for all beneficial aspects of the \emph{\algoname} algorithm, such as cuckoo hashing, which creates alternative locations for items that might otherwise be lost in traditional approaches.
Indeed, in our empirical evaluation (\S\ref{sec:eval}), we demonstrate that CHK's performance consistently exceeds these theoretical guarantees.
\color{black}

%% file: sections/6-CHK-parallel.tex
\section{Concurrent operations}
\label{sec:parallelCHK}

As query operations need to execute while insertions are happening, the problem of concurrent queries and insertions poses challenging questions regarding the associated synchronization.
As outlined in the introduction, there is growing interest in the respective issues.
In this section, we outline the main results along with the algorithmic design space in synchronization of parallel operations, considering multi-threaded, shared-memory systems.

\subsection{Parallel designs and trade-offs}
\begin{table}[tbh!]
    \captionsetup{font=small, skip=2.7pt}
    \caption{Comparison of Parallel Design Categories}
    \label{tab:parallel-designs}
    \footnotesize
    \setlength{\tabcolsep}{2pt}
    \begin{tabularx}{\linewidth}{|l|X|X|X|X@{}}
    \hline
    \textbf{Design} & \makecell{\textbf{Tentative}\\\textbf{scalability}} & \makecell{\textbf{Tentative}\\\textbf{f-query rate}} & \makecell{\textbf{Tentative}\\\textbf{hh-query rate}} \\
    \hline
    Single-shared & Low & High & High \\
    Thread-local & High & - & - \\
    Peer-collaborative & High & Medium/High & Low/Medium \\
    Global-collaborative & Medium/High & - & High  \\
    \emph{Hybrid Peer-Global} & Medium/High & Medium/High & High  \\
    \hline
    \end{tabularx}
\end{table}
\vspace{0.6em}

For parallel processing in shared memory systems, several main design options exist, with properties as summarized in the paragraphs below and sketched in Table~\ref{tab:parallel-designs}.

\subsubsection{Single-shared.} 
Such a design features a single shared-memory data structure, accessible by all threads for insertions and queries.
Insert operations require synchronization mechanisms like locks or atomic operations or helping mechanisms, as in \algoemph{COTS}~\cite{das2009cots}.
In highly parallel environments with high-rate input streams, this design poses challenges regarding scaling with the number of threads.
However, queries only need to access a single data structure, potentially leading to faster responses~\footnote{Note that this depends on the synchronization method. For example, if a reader-writer lock with priority to the writers is employed, a query can starve.}.

\subsubsection{Thread-local.} 
The thread-local design assigns each thread its own local data structure. Threads insert items directly into their respective structures without synchronization, potentially facilitating scalability regarding insertions. However, this approach requires querying all thread-local data structures to collect heavy hitters, which can be inefficient for high-performance scenarios. For example, \algoemph{Topkapi}~\cite{Mandal2018TopkapiPA} implements this approach but does not support concurrent insertions/queries, instead only allowing querying at the end and requiring mergeability~\cite{agarwal2013mergeable} (ability to combine multiple sketches into one without losing accuracy), highlighting the challenge of efficiently aggregating results from multiple sources.

\subsubsection{Peer-collaborative.} 
The peer-collaborative design~\cite{Stylianopoulos2020DelegationSA} enhances the thread-local approach, through \textit{domain-splitting} and \textit{delegated operations}. Each thread is responsible for a subset of items from the universe --- a concept known as \textit{domain-splitting}. If a thread receives operations associated with another thread's domain, it buffers them and delegates these operations accordingly. This method, employed by \algoemph{QPOPSS}~\cite{jarlow2024qpopss}, maintains good scalability even with concurrent insertions and f-queries. However, hh-query consistency is relaxed, potentially introducing bounded staleness. Additionally, hh-queries still require scanning all thread-local data structures, which can introduce higher hh-query latency when the number of threads increases or when hh-queries are frequent.

\subsubsection{Global-collaborative.}
The global-collaborative design combines elements of the single-shared and thread-local approaches. Each thread periodically flushes its local heavy hitters into a single-shared data structure. This allows hh-queries to be answered quickly by accessing one location. Synchronization is less costly compared to the continuous synchronization required in the single-shared approach, since it is not needed at every update. An example of this design is \algoemph{PRIF}~\cite{Zhang2014AnEF}, although it is noteworthy that \algoemph{PRIF} permits only one dedicated thread for hh-queries and does not support f-queries, highlighting the challenges of efficiently handling both types of queries in a global-collaborative setup.

\subsection{Parallel \algoname}
\begin{table}[tbh!]
    \captionsetup{font=small, skip=2.7pt}
    \caption{Additional Notation for Parallel Versions}
    \centering
    \footnotesize
    \begin{tabularx}{\linewidth}{|l|X|}
    \hline
    \textbf{Notation} & \multicolumn{1}{|c|}{\textbf{Description}}  \\
    \hline
    $P$ & Number of parallel threads \\
    $tid$ & Thread identifier (0 to $P-1$) \\
    $ctid$, $otid$ & Current thread ID and owner thread ID for incoming item $e$ \\
    $MAX\_BUF$ & Maximum buffer size before flush \\
    $MAX\_W$ & Maximum allowed weight for buffered items \\
    $owner(e)$ & Thread assignment function $hash(e) \bmod P$ \\
    $CHK_{tid}$ & Thread $tid$'s CHK instance \\
    $B_{tid}[j]$ & Buffer of $\langle e,w \rangle$ pairs from thread $tid$ to thread $j$ \\
    $Q_{tid}$ & Queue of buffer references for thread $tid$ \\
    $PQ_{tid}[j]$ & f-Query slot $\langle e, count, flag \rangle$ from thread $j$ to $tid$ \\
    $HH$ & Global concurrent hash table of heavy hitters \\
    $N_{processed}$ & Global atomic processed stream size counter \\
    \hline
    \end{tabularx}
    \label{tab:parallel_notation}
\end{table}
\vspace{0.6em}
Studying the aforementioned parallel designs, we note multi-faceted trade-offs in terms of parallel scalability and the associated potential for f-query and hh-query efficiency. 
Hence, we seek to balance and improve upon these aspects.
To this end, we propose a general algorithmic framework for parallelizing heavy-hitter detection, with two specific designs tailored to different contexts:

\begin{itemize}[leftmargin=*]
 \item\emph{\paralgonamei}, based on the peer-collaborative design and optimized for situations where insertions and f-queries are predominant.
\item \emph{\paralgonameq}, a hybrid peer-global collaborative approach, combining elements of both peer-collaborative and global-collaborative designs, suited for scenarios where hh-queries are more frequent.
\end{itemize}

Although we use \emph{\algoname} as the underlying algorithm in the description, our parallel designs operate as a \emph{wrapper}, compatible with any heavy-hitter algorithm.
Those with native weighted update support are directly applicable, yet any other algorithm can be used by falling back to repeated unweighted updates.

We now describe first how the two algorithms perform insertion-delegations and f-queries, extending~\cite{Stylianopoulos2020DelegationSA}.
We then explore the differences between them when it comes to hh-queries. In Section~\ref{sec:eval}, we evaluate them both regarding scalability while supporting concurrent insertions, f-queries, and hh-queries, discussing the balancing properties regarding the aforementioned trade-offs. 

\vspace{0.6em}
\input{sections/pseudo-code/mCHK-insertions.tex}

\subsubsection{Insertions and f-queries}
Consider the notation in Table~\ref{tab:parallel_notation}. 
Note that subscripts (e.g., $CHK_{tid}$, $B_{tid}, PQ_{tid}$) indicate thread-local data structures that are independently allocated, which prevents false sharing.
MAX\_BUF should limit buffer size to e.g., fit within one cache line for efficient transfers between threads, while MAX\_W caps the maximum weight per buffered item to balance accuracy (cf. Theorem~\ref{thm:parallel_error_bound}).
Insertions and f-query are implemented as follows:

\noindent\textbf{Insertion} (Alg.~\ref{alg:parallel-chk}, function \hyperref[fn:par-update]{\Update}):
Insertions are delegated when thread $ctid$ receives items $e$ owned by thread $otid$.
Instead of immediate delegating, items are buffered in $B_{ctid}[otid]$.
When the buffer size $|B_{ctid}[otid]| \geq MAX\_BUF$ or the buffer per item $B_{ctid}[otid][e]$ $\geq MAX\_W$, thread $ctid$ adds a reference to the buffer into $Q_{otid}$.
Here, $Q_{tid}$ is a lock-free queue (e.g., LCRQ ~\cite{morrison2013fast}) that stores references to buffers needing processing by thread $tid$.
Delegation operations require the underlying heavy-hitter data structure to be able to handle \emph{weighted updates}, which is supported by the \emph{\algoname} algorithm (cf. \S\ref{sec:weighted_update}). 
When the buffer $B_{ctid}[otid]$ is processed by \hyperref[fn:par-process-updates-q]{\ProcessPendingUpdates}, the thread $otid$ updates its local $CHK_{ctid}$ and increments the global atomic stream size counter $N_{processed}$ by $w$.
Note that the \hyperref[fn:par-process-updates-q]{\ProcessPendingUpdates} implementation differs between the two parallel designs: in \emph{\paralgonameq} (Alg.\ref{alg:parallel-chk-qo}), it additionally identifies items that exceed the heavy-hitter threshold and adds them to a global hash table $HH$, whereas \emph{\paralgonamei} (Alg.\ref{alg:parallel-chk-io}) only performs the local updates.

\noindent\textbf{f-query} (Alg.~\ref{alg:parallel-chk}, function \hyperref[fn:par-query]{\Query}):
Similar to updates, f-queries from thread $ctid$ to $otid$ are also delegated through pending f-query slots ($PQ_{otid}[ctid]$).
Each slot stores a tuple $\langle e, count, flag \rangle$ where $e$ is the queried item, $count$ stores the result, and $flag$ indicates processing status.
When thread $ctid$ needs to f-query an item $e$ owned by thread $otid$, it initializes a slot with $\langle e, 0, unprocessed \rangle$ in $PQ_{otid}[ctid]$.
The querying thread $ctid$ monitors the $flag$ in its assigned slot until it changes to $processed$.
At the same time, instead of waiting idly, the querying thread $ctid$ continuously processes its own pending updates and queries.
Meanwhile, delegated thread $otid$ processes the f-query by querying the frequency from its local $CHK_{otid}$.
Once processed, $otid$ updates the slot with the final count and marks the $flag$ as $processed$.
This design allows continuous f-querying without blocking or freezing thread execution.

\vspace{0.6em}
\input{sections/pseudo-code/mCHK-I.tex}
\input{sections/pseudo-code/mCHK-q.tex}
\subsubsection{hh-queries.} 
\hfill\\
\noindent\textbf{\paralgonamei} (Alg.~\ref{alg:parallel-chk-io}): 
When a hh-query is executed, the algorithm performs a non-blocking scan across threads' local $CHK_{tid}$ structures to collect items whose counts exceed the threshold ($count \geq \phi N_{processed} $) into the result set $\hat{R}$. For thread safety, mCHK-I uses opportunistic thread-level locking --- if a thread's lock cannot be acquired immediately, the algorithm continues scanning other threads and processes pending updates, before retrying locked threads later. As this is a low-contention locking when the hh-query rate is not too high, it potentially does not cause a high number of retries.

\noindent\textbf{\paralgonameq} (Alg.~\ref{alg:parallel-chk-qo}): 
It improves hh-query latency and overall throughput under frequent hh-queries by maintaining a global concurrent hash table $HH$ (e.g., libcuckoo~\cite{li2014algorithmic}) for heavy hitters. When thread $ctid$ processes updates (Alg.~\ref{alg:parallel-chk-qo}, \hyperref[fn:par-process-updates-q]{\ProcessPendingUpdates}), items exceeding the threshold ($count \geq \phi N_{processed} $) are added to $HH$. Although this introduces synchronization overhead, the cost is amortized over multiple updates due to buffering. 
When the hh-query is executed, the algorithm performs a non-blocking scan of $HH$ to collect items whose counts exceed the threshold ($count \geq \phi N_{processed} $), using double-collecting~\cite{afek1993atomic}
(Alg.~\ref{alg:parallel-chk-qo}, l. 13-17, repeatedly collecting the data twice until values match)
to avoid torn reads and adding them to the result set $\hat{R}$.

\subsection{Accuracy of hh-queries}
Consider the global state of the algorithm at a point in time, applicable to both \emph{mCHK-I} and \emph{mCHK-Q}.
Let $N_s$ be the total weighted size of the processed data stream at the start of the query, and  $\mathbb{B}(e)$ be the total buffered weight of item $e$ i.e., the weight of $e$ that has not yet been processed by the algorithm ($\mathbb{B}(e) \leq P \times MAX\_W$, since each thread can buffer at most $MAX\_W$ weight for any item).

\begin{theorem}[Parallel Approximation Bound]
\label{thm:parallel_error_bound}
Let $e$ be any item in the heavy part, $f_{N_s}(e)$ be the true frequency of $e$ up to $N_s$; the estimated frequency $\hat{f}(e)$ from hh-query satisfies:
\[
    \Pr\left[ |\hat{f}(e) - \left(f_{N_s}(e) -  \mathbb{B}(e)\right)| \geq  \epsilon N_s \right] \leq \frac{1}{\epsilon \mathcal{B}}
\]
\end{theorem}

\begin{proof}
At query time, the underlying sequential \emph{\algoname{}} has processed $f_{N_s}(e) - \mathbb{B}(e)$ weight for item $e$. Applying theorem~\ref{thm:approximation_bounds} with this adjusted frequency yields the result.
This implies an important practical trade-off: smaller $MAX\_W$ and $MAX\_BUF$ values reduce buffering delay (improving accuracy) but increase synchronization frequency (reducing throughput).
\end{proof}

%% file: sections/pseudo-code/mCHK-insertions.tex
\setlength{\algomargin}{0.2em}
\begin{algorithm}
  \caption{Parallel {\algoname} Wrapper}
  \label{alg:parallel-chk}
  \footnotesize 
    \setstretch{0.8} 
  \SetKwProg{Fn}{Procedure}{}{}
  \SetKwFunction{Update}{Update}
  \SetKwFunction{Query}{f-Query}
  \SetKwFunction{ProcessPendingUpdates}{ProcessPendingUpdates}
  \SetKwFunction{ProcessPendingQueries}{ProcessPendingf-Queries}
  \SetKwFunction{hh-Query}{hh-Query}
  
  \Fn{\Update{$e, w$}}{{\label{fn:par-update}}
      $ctid$, $otid \gets$ current thread ID, owner(e) \inlinetcp{Domain splitting}\;
      Add $\langle e, w \rangle$ to $B_{ctid}[otid]$ \inlinetcp{Buffer for delegated processing}\;
      \If{$|B_{ctid}[otid]| \geq MAX\_BUF$ or  $B_{ctid}[otid][e] \geq MAX\_W$}{
          Add reference of $B_{ctid}[otid]$ to $Q_{otid}$ \inlinetcp{Delegate to otid}\;
          \While{$B_{ctid}[otid]$ not processed}{
              \hyperref[fn:par-process-updates-i]{\ProcessPendingUpdates} \inlinetcp{Process work while waiting}\; 
          }
      }
  }
  
  \Fn{\Query{$e$}}{{\label{fn:par-query}}
      $ctid \gets$ current thread ID\;
      $otid \gets owner(e)$\;
      Store f-query $e$ in slot $PQ_{otid}[ctid]$ \inlinetcp{Delegate query to otid}\;
      \While{f-query in $PQ_{otid}[ctid]$ not processed}{
          \hyperref[fn:par-process-updates-i]{\ProcessPendingUpdates} \inlinetcp{Process work while waiting}\; 
          \hyperref[fn:par-process-queries]{\ProcessPendingQueries{}}\;
      }
      \Return result from $PQ_{otid}[ctid]$\;
  }
  
  \Fn{\ProcessPendingQueries}{{\label{fn:par-process-queries}}
      $ctid \gets$ current thread ID\;
      \For{$tid \gets 0$ \KwTo $P-1$}{
          \If{$PQ_{ctid}[tid]$ has unprocessed f-query $e$}{
              $count \gets CHK_{ctid}.$\hyperref[alg:chk-main]{\Query{$e$}} \inlinetcp{Process query}\;
              Store $count$ as result in $PQ_{ctid}[tid]$\;
              Mark f-query in $PQ_{ctid}[tid]$ as processed\;
          }
      }
  }

\end{algorithm}

%% file: sections/pseudo-code/mCHK-I.tex
\setlength{\algomargin}{0.2em}
\begin{algorithm}
  \caption{{\paralgonamei} operations}
  \label{alg:parallel-chk-io}
  \footnotesize
\setstretch{0.8} 
  \SetKwProg{Fn}{Procedure}{}{}
  \SetKwFunction{QueryHeavyHitters}{hh-Query}
  
  \Fn{\ProcessPendingUpdates}{{\label{fn:par-process-updates-i}}
      $ctid \gets$ current thread ID\;
      \If{cannot acquire lock on $ctid$'s data}{
          \Return\;
      }
      \While{$Q_{ctid}$ has unprocessed references}{
          $B_{ref} \gets$ get next unprocessed reference from $Q_{ctid}$\;
          \ForEach{$\langle e, w \rangle \in B_{ref}$}{
              $CHK_{ctid}.$\hyperref[alg:chk-main]{\Update{$e, w$}} \inlinetcp{Delegated updates}\;
              Atomic add $w$ to $N_{processed}$ \inlinetcp{Track stream size}\;
          }
          Mark $B_{ref}$ as processed in $Q_{ctid}$\;
      }
      release lock on thread $ctid$'s data\;
  }

    \Fn{\QueryHeavyHitters}{{\label{fn:par-query-heavy-i}}
    $\hat{R}, remaining, made\_progress \gets \emptyset, P, false$\;
    \While{$remaining > 0$}{
        $made\_progress \gets false$\;
        \For{$tid \gets 0$ \KwTo $P-1$}{
            \If{$tid$'s data is not scanned and can acquire lock}{
                $cand \gets CHK_{tid}$.hh-Query()\;
                \ForEach{$\langle e, count \rangle \in cand$}{
                    \If{$count \geq \phi N_{processed} $}{
                        Add $\langle e, count \rangle$ to $\hat{R}$\;
                    }
                }
                Mark thread $tid$'s data as scanned and release lock\;
                $remaining, made\_progress \gets remaining - 1, true$\;
            }
        }
        \If{$remaining > 0$ and not $made\_progress$}{
            \ProcessPendingUpdates\; 
        }
    }
    \Return $\hat{R}$\;
    }
\end{algorithm}

%% file: sections/pseudo-code/mCHK-q.tex
\setlength{\algomargin}{0.2em}
\begin{algorithm}
  \caption{{\paralgonameq} operations}
  \label{alg:parallel-chk-qo}
  \footnotesize
    \setstretch{0.8} 
  \SetKwProg{Fn}{Procedure}{}{}
  \SetKwFunction{ProcessPendingUpdates}{ProcessPendingUpdates}
  \SetKwFunction{QueryHeavyHitters}{hh-Query}
  
   \Fn{\ProcessPendingUpdates}{{\label{fn:par-process-updates-q}}
      $ctid \gets$ current thread ID\;
      \While{$Q_{ctid}$ has unprocessed references}{
          $B_{ref} \gets$ get next unprocessed reference from $Q_{ctid}$\; 
          \ForEach{$\langle e, w \rangle \in B_{ref}$}{
              $count \gets$ $CHK_{ctid}.$\hyperref[alg:chk-main]{\Update{$e, w$}} \inlinetcp{Delegated updates}\;
              Atomic add $w$ to $N_{processed}$ \inlinetcp{Track stream size}\;
              \If{$count \geq \phi N_{processed} $}{
                  Update $\langle e, count \rangle$ in $HH$ \inlinetcp{Maintain global HH}\;
              }
          }
          Mark $B_{ref}$ as processed\;
      }
  }

  \Fn{\QueryHeavyHitters}{{\label{fn:par-query-heavy-q}}
  
  $\hat{R} \gets \emptyset$\;
    \tcp*[h]{HH: Global concurrent hash table of heavy hitters}
    
      \ForEach{entry position $i$ in $HH$}{
        \Repeat{$e_1 = e_2$ and $count_1 = count_2$ \inlinetcp{double collecting}}{
            $\langle e_1, count_1 \rangle \gets HH[i]$ \inlinetcp{First read}\;
            $\langle e_2, count_2 \rangle \gets HH[i]$ \inlinetcp{Second read for consistency}\;
        }
        \If{$count_1 \geq \phi N_{processed}$}{
            Add $\langle e_1, count_1 \rangle$ to $\hat{R}$\;
        }
        }
      \Return $\hat{R}$\;
  }

\end{algorithm}

%% file: sections/7-evaluation.tex
\section{Evaluation} 
\label{sec:eval}

This section presents a comprehensive evaluation of the contributed methods, organized in two parts: (1)~evaluation of  \emph{CHK} on both real-world  and synthetic data, compared to state-of-the-art algorithms, focusing on accuracy and throughput, with varying memory constraints and data distributions (skewness); 
(2)~evaluation of our parallel designs on varying hardware features, studying scalability with thread counts and hh-query rates; the study includes both \emph{CHK} and alternative underlying heavy-hitter detection algorithms.
We begin by describing our experiment setup, including hardware specifications, datasets, and evaluation metrics.

\noindent\textbf{Platforms:} 
We conducted experiments on two hardware platforms. Table~\ref{tab:platforms} provides the specifications:

\begin{table}[tbh]
  \captionsetup{font=small, skip=2.7pt}
  \caption{Hardware Platform Specifications}
  \label{tab:platforms}
  \footnotesize
  \setlength{\tabcolsep}{4pt}
  \begin{tabularx}{\linewidth}{|l|X|X|}
  \hline
  & \makecell{\textbf{Platform A}} & \makecell{\textbf{Platform B}} \\
  \hline
  \textbf{Processor} & \makecell{Intel(R) Xeon(R) E5-2695 v4\\(NUMA dual sockets)} & \makecell{AMD EPYC 9754\\(UMA single socket)} \\
  \hline
  \textbf{Cores/Clock} & \makecell{36 cores/2.1 GHz} & \makecell{128 cores/2.25 GHz} \\
  \hline
  \textbf{Hyper Threading} & \makecell{Enabled} & \makecell{Disabled} \\
  \hline
  \textbf{Cache (L1/L2/L3)} & \makecell{32KB/256KB/45MB} & \makecell{32KB/1024KB/256MB} \\
  \hline
  \textbf{Used for} & \makecell{All experiments}  & \makecell{Parallel experiments only} \\
  \hline
  \end{tabularx}
\end{table}

\noindent\textbf{Data sets:} 
We used three datasets: 
(1) \textit{CAIDA\_L}, a real-world dataset of source IPs (skewness $\sim$ 1.2); 
(2) \textit{CAIDA\_H}, a real-world dataset of source ports (skewness $\sim$ 1.5).
Both are derived from the CAIDA Anonymized Internet Traces 2018 ~\cite{caida2018passive}, a broadly used benchmark in heavy-hitter detection literature~\cite{Yang2019HeavyKeeperAA, Yang2018ElasticSA, shi2023cuckoo}. From these traces, we extracted source IPs and source ports from the first 10M packets, representing network traffic monitoring scenarios faced in real-world environments with different skewness.
(3) \textit{Synthetic data} containing 10M items generated using Zipf distributions with skewness $\alpha$ ranging from 0.8 to 1.6, commonly used to model real-world frequency distributions~\cite{Newman01092005, Stylianopoulos2020DelegationSA}.

\noindent\textbf{Baselines:}
We compared \emph{CHK} against representative state-of-the-art heavy-hitter detection algorithms: \algoemph{Space-Saving} (SS) ~\cite{Metwally2006AnIE}, \algoemph{Count-MinSketch} (CMS)~\cite{Cormode2004AnID}, \algoemph{AugmentedSketch} (AS) ~\cite{roy2016augmented}  and \algoemph{HeavyKeeper} (HK)~\cite{Yang2019HeavyKeeperAA}, all implemented in C++, compiled with -O2 and available in our repository~\cite{ngo_2025_15593109}.
CHK uses parameters from \hyperlink{box:param-config}{\textbf{Recommended Parameter Configuration}} (cf. \S\ref{sec:eval_seq_methodology} for experiment setup), while others follow their prescribed recommendations.
All algorithms use auxiliary heaps for continuous heavy-hitter tracking (cf. \S\ref{sec:chk_main_operations}), as needed in real-world monitoring scenarios.

\hypertarget{box:param-config}{}
\tcolorbox[
  colback=white,
  colframe=MaterialOrange400,
  fonttitle=\small\bfseries,
  fontupper=\small,
  left=1pt,
  right=1pt,
  top=1pt,
  bottom=1pt,
  toptitle=0pt,       
  bottomtitle=0pt,    
  titlerule=0.4pt,    
  toprule=0.4pt,      
  title=Recommended Parameter Configuration,
]

\textbf{Bucket Configuration:} Fig.~\ref{fig:parameter_sensitivity} shows that 2 heavy entries per bucket achieve better accuracy, which aligns with prior research on cuckoo hashing~\cite{Dietzfelbinger2005BucketizedCuckoo,Fan2014CuckooFP,Fountoulakis2011CuckooProof} showing 2-4 entries maximizes space efficiency. We use a promotion threshold $L=16$ based on sensitivity analysis in Fig.~\ref{fig:parameter_sensitivity}. For counter sizes, we use 8-bit counters in the \textit{lobby part} (sufficient for $L$) and 32-bit counters in the \textit{heavy part} (to track frequencies of heavy items). This also ensures each bucket fits within a CPU cache line for optimal memory access performance.
\textbf{Other Parameters:} 16-bit fingerprints following~\cite{Fan2014CuckooFP} and decay factor $b=1.08$ following~\cite{Yang2019HeavyKeeperAA}.

\endtcolorbox
\begin{figure}[tbh]
  \centering
  \includegraphics[width=\linewidth]{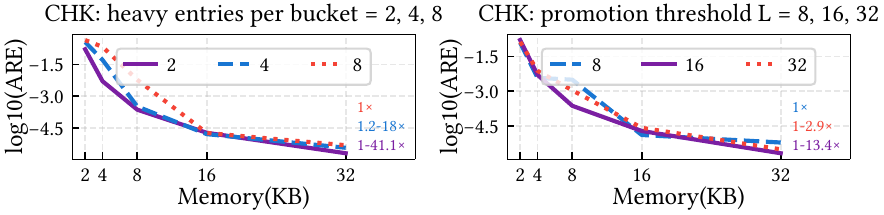}
  \captionsetup{skip=2pt, font=small}
  \caption{\vinhrevise{}{Sensitivity to decay base and heavy entries per bucket}}
  \label{fig:parameter_sensitivity}
  \Description{\emph{\algoname} sensitivity to decay base and heavy entries per bucket}
\end{figure}

\subsection{Study of the sequential algorithms}

\subsubsection{Measurement Methodology}
\label{sec:eval_seq_methodology}
We use four metrics: precision, recall, ARE, and throughput as defined in \S\ref{sec:prelim}.
We vary different parameters depending on the dataset type.
With synthetic data, we vary $\phi$, memory usage, and skewness $\alpha$, modifying one at a time.
For the CAIDA datasets, where skewness is an inherent characteristic, we vary only $\phi$ and memory usage.
When parameters are not varying, their values were set to: threshold $\phi$ of 0.0005, memory usage of 4KB, and skewness of 1.2, with the following reasoning:
the skewness is a value in the range of real-world network traffic distribution~\cite{caida2018passive}, without favoring algorithms optimized for either high or low skew; 
$\phi$ corresponds to a rather low threshold, i.e., with low selectivity, matching cases of many heavy hitters, while memory is set to be just adequate to keep them.
This enables evaluating the algorithms under limited memory rather than simply testing performance under abundant resources. Besides, memory size matters when the information is to be communicated (cf. e.g.~\cite{harris2023compressing}), so the smaller the sketch, the better.
Each experiment was run 30 times, with average performance calculated and plotted. 
To enable easier comparisons, we also reported \emph{relative improvements} rather than just plain absolute values. 
For each configuration, we calculated each algorithm's improvement normalized over the least-performing one and presented these numbers in sorted order, based on the average of point-wise improvement ratio.

\begin{figure}[tbh!]
  \centering
  \includegraphics[width=\linewidth]{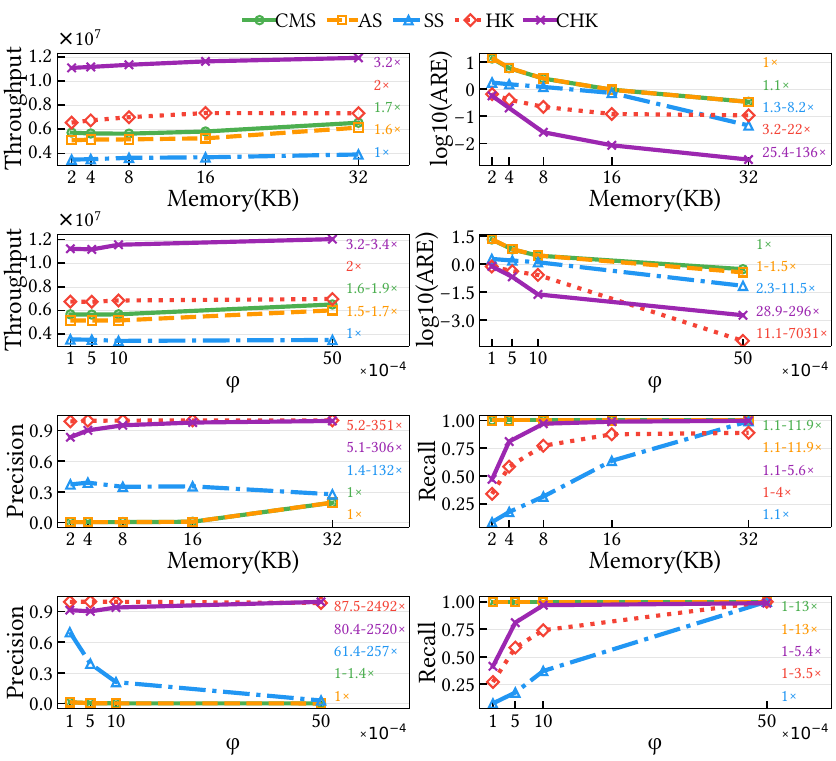}
  \captionsetup{skip=2pt, font=small}
  \caption{
  Plots show performance of sequential algorithms on \textit{CAIDA\_L} dataset for throughput, $log_{10}$(ARE), precision, and recall across varying memory, and $\phi$. See \S\ref{sec:eval_seq_methodology} for calculation details.
  }
  \label{fig:caida_l}
  \Description{
  Plots show performance of sequential algorithms on \textit{CAIDA\_L} dataset for throughput, $log_{10}$(ARE), precision, and recall across varying memory, and $\phi$. See \S\ref{sec:eval_seq_methodology} for calculation details.
  }
\end{figure}

\setlength{\floatsep}{0pt} 
\begin{figure}[tbh!]
  \centering
  \includegraphics[width=\linewidth]{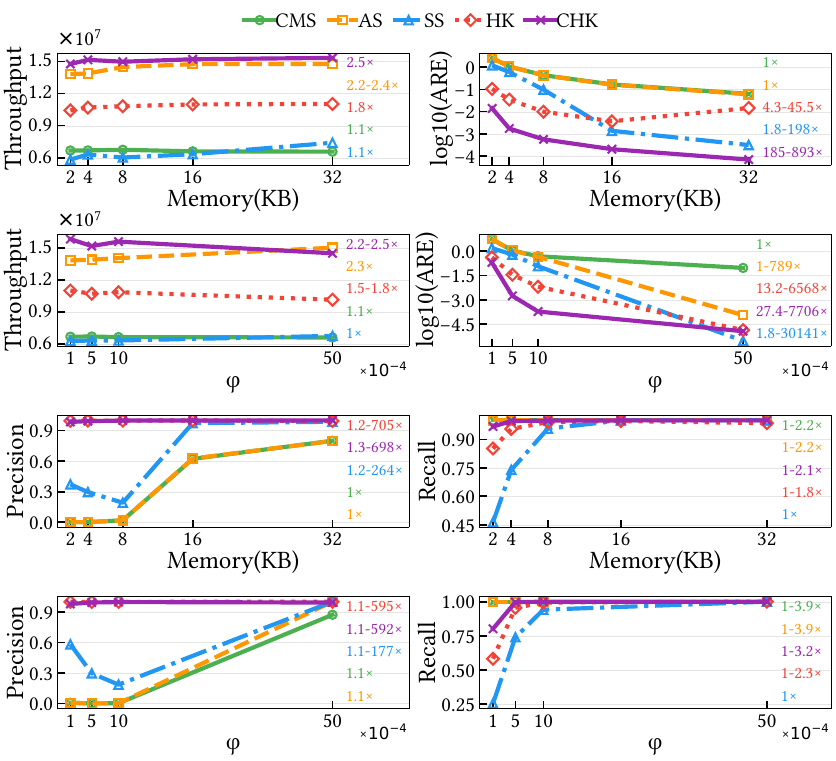}
  \captionsetup{skip=2pt, font=small}
  \caption{
  Plots show performance of sequential algorithms on \textit{CAIDA\_H} dataset for throughput, $log_{10}$(ARE), precision, and recall across varying memory, and $\phi$. See \S\ref{sec:eval_seq_methodology} for calculation details.
  }
  \label{fig:caida_h}
  \Description{Plots show performance of sequential algorithms on \textit{CAIDA\_H} dataset for throughput, $log_{10}$(ARE), precision, and recall across  varying memory, and $\phi$. See \S\ref{sec:eval_seq_methodology} for calculation details.}
\end{figure}

\subsubsection{Experiments on throughput}
\label{sec:eval_seq_throughput}
In streaming heavy-hitter algorithms, throughput is influenced by configuration parameters and dataset characteristics. 
Limited memory increases collisions, requiring more hash resolution operations; low skewness leads to more items competing for slots; and low thresholds classify more items as heavy hitters, increasing heap maintenance overhead. 
These factors reduce throughput by triggering more expensive algorithm paths.
Our results across synthetic and real-world datasets (Fig.~\ref{fig:caida_l}, Fig.~\ref{fig:caida_h}, and Fig.~\ref{fig:zipf}) confirm these patterns
and show the benefits of the CHK design, which is summarized in \hyperlink{box:key-takeaway1}{Key Takeaway 1}.

\hypertarget{box:key-takeaway1}{}
\tcolorbox[
  colback=white,
  colframe=MaterialOrange400,
  fonttitle=\small\bfseries,
  fontupper=\small,
  left=1pt,
  right=1pt,
  top=1pt,
  bottom=1pt,
  toptitle=0pt,       
  bottomtitle=0pt,    
  titlerule=0.4pt,    
  toprule=0.4pt,      
  title=Key Takeaway 1 -- on sequential CHK throughput,
]
CHK consistently delivers superior throughput across all tested configurations, by 2-3 times in most settings. 
While some methods, e.g., AS, demonstrate occasional performance spikes under high skew conditions (when few items dominate the traffic), they still underperform compared to CHK slightly, even in these favorable scenarios, and degrade substantially under low skew distributions.
The improvements are possible through CHK's inverted filter (lobby) process and its enabled system-aware layout, where data movement is limited (and often within the same cache line), and hash collision resolution is applied selectively only to heavy-hitter candidates.
As a result, CHK generalizes across varied workloads and offers more predictable performance in diverse scenarios, where stream characteristics cannot be known in advance and memory requirements are difficult to determine ahead of time.
\endtcolorbox

\subsubsection{Experiments on accuracy}
Experiments on accuracy (precision, recall, and ARE of Fig.~\ref{fig:caida_l}, Fig.~\ref{fig:caida_h}, and Fig.~\ref{fig:zipf}) show varying behavior across different algorithms and configurations. 
In general, each algorithm shows similar improvement trends when varying parameters (for reasons similar to those discussed in \S\ref{sec:eval_seq_throughput}), but the magnitude of improvement differs. 
At low memory, low skewness, or low threshold settings, algorithms exhibit distinct trade-offs: CMS/AS/SS achieve high recall but poor precision because they overestimate frequencies indiscriminately, which allows them to find most heavy hitters but results in many false positives.
HK achieves high precision but lower recall as multiple heavy items can be mapped to the same bucket, causing collisions and omissions of true heavy hitters.
CHK can balance, maintaining higher precision than CMS/AS/SS while delivering better recall than HK in most settings, as explained in \hyperlink{box:key-takeaway2}{Key Takeaway 2}.

\hypertarget{box:key-takeaway2}{}
\tcolorbox[
  colback=white,
  colframe=MaterialOrange400,
  fonttitle=\small\bfseries,
  fontupper=\small,
  left=1pt,
  right=1pt,
  top=1pt,
  bottom=1pt,
  toptitle=0pt,       
  bottomtitle=0pt,    
  titlerule=0.4pt,    
  toprule=0.4pt,      
  title=Key Takeaway 2 -- on sequential CHK accuracy,
]

Under constrained conditions (low memory usage or low skewness), CHK
improves accuracy by a large margin. This is because CHK uses cuckoo collision resolution, which allows it to capture more heavy hitters, resulting in better recall, while maintaining a low false positive rate even under stringent memory constraints.

\endtcolorbox
\begin{figure}[tbh!]
  \centering
  \includegraphics[width=\linewidth]{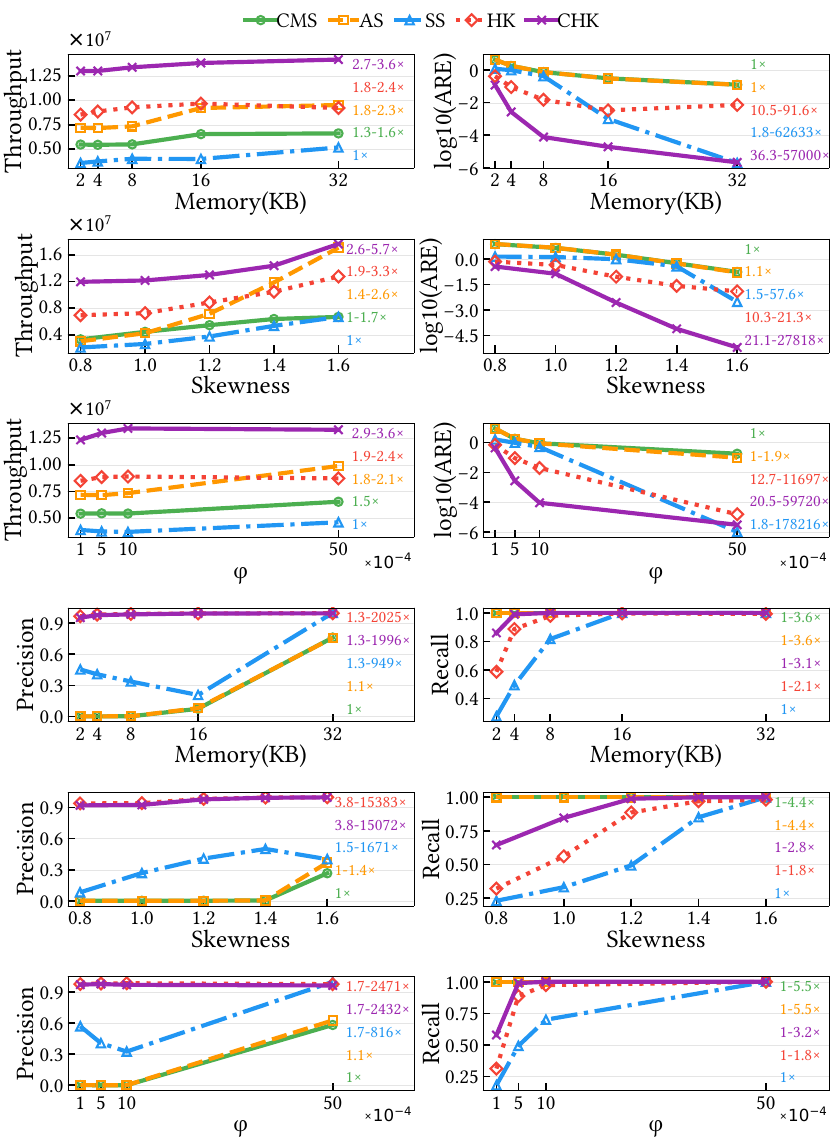}
  \captionsetup{skip=2pt, font=small}
  \caption{
  Plots show performance of sequential algorithms on \textit{Synthetic data} for throughput, $log_{10}$(ARE), precision, and recall across varying skewness, memory, and $\phi$. See \S\ref{sec:eval_seq_methodology} for calculation details.}
  \label{fig:zipf}
  \Description{
  Plots show performance of sequential algorithms on \textit{Synthetic data} for throughput, $log_{10}$(ARE), precision, and recall across varying skewness, memory, and $\phi$. See \S\ref{sec:eval_seq_methodology} for calculation details.
  }
\end{figure}

\subsection{Study of the parallel algorithms}
\subsubsection{Measurement Methodology}
We evaluated our parallel framework across multiple dimensions: hardware platforms (\textit{Platform A} and \textit{B}), underlying heavy-hitter detection algorithms, thread counts, and hh-query rates (ratio of hh-queries vs total operations performed by each thread),
for insertion throughput and hh-query latency (\S\ref{sec:prelim}). 
All experiments used the \textit{CAIDA\_H} dataset.
Since both variants handle f-queries similarly via the \textit{delegation mechanism}, we focus the comparison on their divergent approaches to hh-queries.
To evaluate the framework's generality, we used various underlying algorithms that natively support weighted updates (CHK, AS, CMS, SS; HK was not included since it lacks this feature). 
The resulting parallel variants are denoted as \emph{\paralgonamei}, \emph{\paralgonameq}, mAS-I, mAS-Q, mCMS-I, mCMS-Q, mSS-I, and mSS-Q, with -I and -Q indicating insertion- and query-optimization.
With similar reasoning as in subsection~\ref{sec:eval_seq_methodology}, $\phi$ = 0.00005 and memory = 1KB per thread, with algorithm parameters $MAX\_W$ = 1000 and $MAX\_BUF$ = 16, to simulate a resource-constrained environment.
Each experiment ran 30 times, with results showing both average performance and relative speedup compared to single-thread implementations.

\subsubsection{Experiments on throughput}
Fig.~\ref{fig:par_caida_H_throughput} shows the average throughput of parallel variants with varying thread counts and hh-query rates across different underlying heavy-hitter detection algorithms.
Both designs scale nearly linearly as thread count increases.
When comparing the query-optimized (-Q) and insertion-optimized (-I) variants, we observe patterns consistent with the design trade-offs discussed in \S\ref{sec:parallelCHK}. 
The -I variants outperform their -Q counterparts in insertion-only workloads (0\% query rate) due to their reduced synchronization overhead. However, as query rates increase, the -Q variants demonstrate significantly better performance. 
Notable results are summarized in \hyperlink{box:key-takeaway3}{Key Takeaway 3} and \hyperlink{box:key-takeaway4}{Key Takeaway 4}.

\begin{figure}[tbh!]
  \centering
  \includegraphics[width=\linewidth]{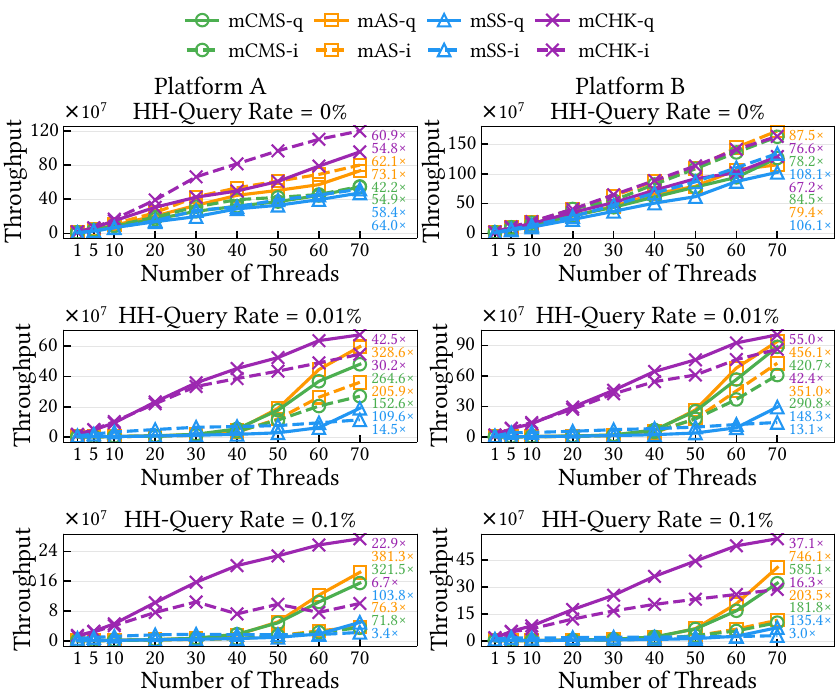}
  \captionsetup{skip=2pt, font=small}
  \caption{
  Plots show throughput and relative improvement compared to single-thread implementation of parallel variants on \textit{CAIDA\_H} across varying hh-query rates, thread counts, and platforms. 
  }
  \label{fig:par_caida_H_throughput}
  \Description{Plots show throughput of parallel variants on \textit{CAIDA\_H} across varying hh-query rates, thread counts, and platforms. }
\end{figure}

\setlength{\floatsep}{5.7pt} 
\hypertarget{box:key-takeaway3}{}
\tcolorbox[
floatplacement=tbh!,float,
  colback=white,
  colframe=MaterialOrange400,
  fonttitle=\small\bfseries,
  fontupper=\small,
  left=1pt,
  right=1pt,
  top=1pt,
  bottom=1pt,
  toptitle=0pt,       
  bottomtitle=0pt,    
  titlerule=0.4pt,    
  toprule=0.4pt,      
  title=Key Takeaway 3 -- on parallel throughput,
]

Both \emph{\paralgonamei} and \emph{\paralgonameq} consistently outperform other parallel variants in throughput by a substantial margin, especially at low to moderate thread counts. This difference arises from CHK's lean operations and accuracy improvements in heavy-hitter estimation, which together reduce the number of false positives that must be transferred during hh-queries. 
These results strengthen the claim that CHK and its parallel designs are adaptive to various workload conditions. Additionally, the performance gains at even lower thread counts suggests potential benefits in energy saving and elasticity possibilities, i.e. better resource utilization.
\endtcolorbox

\hypertarget{box:key-takeaway4}{}
\tcolorbox[
floatplacement=tbh!,float,
  colback=white,
  colframe=MaterialOrange400,
  fonttitle=\small\bfseries,
  fontupper=\small,
  left=1pt,
  right=1pt,
  top=1pt,
  bottom=1pt,
  toptitle=0pt,       
  bottomtitle=0pt,    
  titlerule=0.4pt,    
  toprule=0.4pt,      
  title=Key Takeaway 4 -- on parallel framework portability,
]

Our parallel framework can scale across diverse hardware,  with differences in memory architecture (NUMA, UMA) and processor features (hyperthreading,  non-hyperthreading). 
This portability extends even to cross-socket execution because of the {delegated operations} approach (also in~\cite{Stylianopoulos2020DelegationSA, jarlow2024qpopss, hilgendorf2025lmqsketchlagommultiquerysketch}), enhanced by the weighted update capability, which allows threads to aggregate work and streamline it with synchronization, reducing communication overhead. 
Furthermore, in hyper-threading environments, where sibling threads must share execution resources, CHK outperforms other algorithms by enabling threads to complete more workload in the same amount of time.
\endtcolorbox

\subsubsection{Experiments on hh-query latency}
Fig.~\ref{fig:par_caida_H_latency} shows the hh-query latency of our parallel implementations with varying thread counts and hh-query rates across different underlying algorithms. 
The results differ significantly between algorithms due to their query operation costs and overhead of transferring detected heavy hitters.
When comparing -Q and -I variants, we observe patterns consistent with the design projections. The -I variants exhibit higher latency that increases with thread count due to the need to interact with each thread to access thread-local structures. In contrast, -Q variants typically maintain more stable latency, especially at higher thread counts, due to their centralized query architecture.
However, an interesting exception occurs with CMS and AS variants under constrained memory conditions.
For these algorithms, the -Q variants actually show higher latency than expected, as they must frequently synchronize large numbers of misclassified heavy hitters to the global structure, creating substantial overhead.
The zoomed-in view of the results (with the y-axis scaled to match CHK's latency range) reveals that \emph{\paralgonameq} shows very low and stable query latency compared to other algorithms. After reaching a certain thread count, \emph{\paralgonameq} consistently maintains latency below 150 $\mu$sec even under high thread counts and query rates.

\vspace{0.3em}
\begin{figure}[tbh!]
  \centering
  \includegraphics[width=\linewidth]{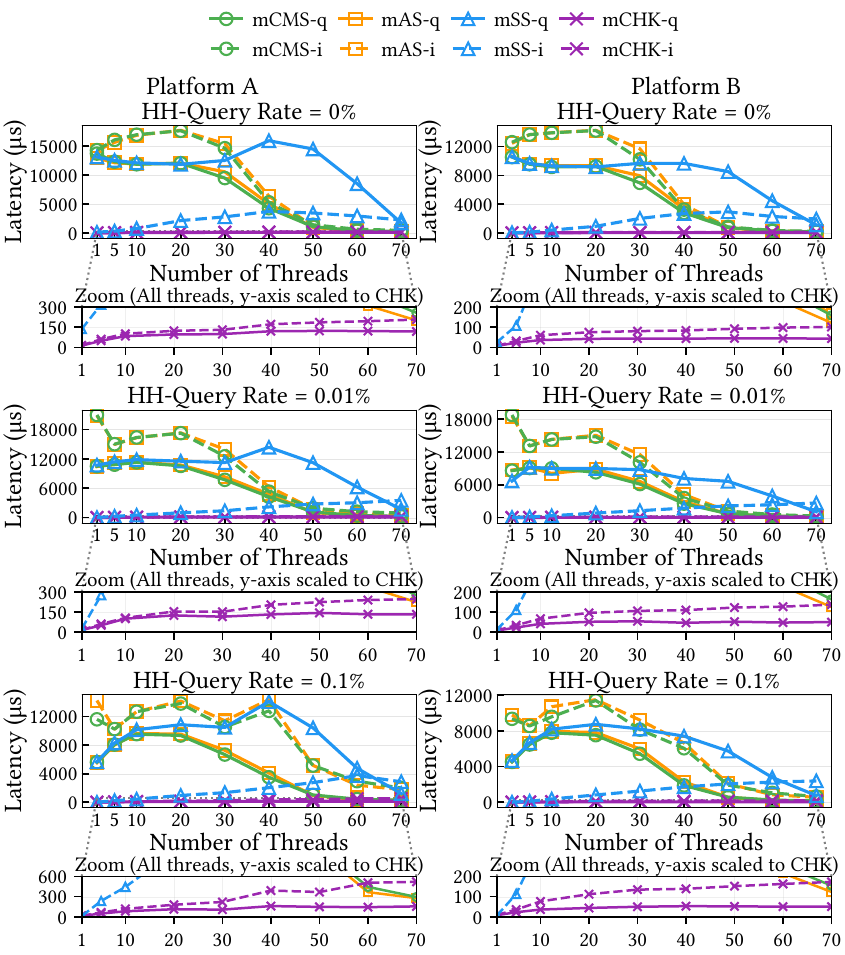}
  \captionsetup{skip=2pt, font=small}
  \caption{
  Plots show hh-query latency of parallel variants on \textit{CAIDA\_H} across varying hh-query rates, thread counts, platforms.
  }
  \label{fig:par_caida_H_latency}
  \Description{
  Plots show hh-query latency of parallel variants on \textit{CAIDA\_H} across varying hh-query rates, thread counts, and platforms.
  } 
\end{figure} 

%% file: sections/10-conclusions.tex
\section{Conclusions and Future Work}
\label{sec:concl}
We introduced \emph{\algoname} (CHK), a fast, accurate, and space-efficient algorithm that delivers orders of magnitude better throughput and accuracy compared to state-of-the-art methods, even with tight memory and low-skew data. For parallel scalability, we proposed \emph{mCHK-I} and \emph{mCHK-Q}, achieving near-linear scaling with low hh-query latencies. 
These parallel algorithms operate as a \emph{wrapper} around any sequential heavy-hitter, without requiring mergeability, which enables modular integration into existing systems with minimal changes. This makes CHK and its parallel variants useful both as standalone algorithmic designs and as integrable building blocks within databases, stream processing engines, and data analytics frameworks. Future research directions include extending \emph{\algoname} to support sliding windows with potential integration into established data processing systems.\looseness=-1